\documentclass[12pt]{article}

\usepackage[margin=1in]{geometry}
\usepackage{amsthm}
\usepackage{amssymb}
\usepackage{amsmath}
\usepackage{graphicx}
\usepackage{mathdots}
\usepackage{mathtools}
\usepackage{url}
\usepackage{color}
\usepackage{framed}
\usepackage{tikz}
\usepackage{stmaryrd}
\usepackage[ruled,vlined]{algorithm2e}
\usepackage{array}

\newcommand{\vsubseteq}{\rotatebox[origin=c]{90}{$\subseteq$}}

\newcommand{\llangle}{\langle\hspace{-2.5pt}\langle}
\newcommand{\rrangle}{\rangle\hspace{-2.5pt}\rangle}

\tikzset{every node/.style={shape=circle,fill=black,inner sep=0pt, minimum size=0.5em,anchor=mid}}

\newtheorem{theorem}{Theorem}
\newtheorem{corollary}[theorem]{Corollary}
\newtheorem{lemma}[theorem]{Lemma}

\theoremstyle{definition}
\newtheorem{problem}[theorem]{Problem}
\newtheorem{example}[theorem]{Example}
\newtheorem{definition}[theorem]{Definition}

\title{Group-invariant max filtering}

\author{
Jameson~Cahill\footnote{Department of Mathematics and Statistics, University of North Carolina Wilmington, Wilmington, NC} 
\quad
Joseph~W.~Iverson\footnote{Department of Mathematics, Iowa State University, Ames, IA} 
\quad
Dustin~G.~Mixon\footnote{Department of Mathematics, The Ohio State University, Columbus, OH} \footnote{Translational Data Analytics Institute, The Ohio State University, Columbus, OH}
\quad
Daniel~Packer\footnotemark[3]
}

\date{}

\begin{document}
\maketitle

\begin{abstract}
Given a real inner product space $V$ and a group $G$ of linear isometries, we construct a family of $G$-invariant real-valued functions on $V$ that we call \textit{max filters}.
In the case where $V=\mathbb{R}^d$ and $G$ is finite, a suitable max filter bank separates orbits, and is even bilipschitz in the quotient metric.
In the case where $V=L^2(\mathbb{R}^d)$ and $G$ is the group of translation operators, a max filter exhibits stability to diffeomorphic distortion like that of the scattering transform introduced by Mallat.
We establish that max filters are well suited for various classification tasks, both in theory and in practice.
\end{abstract}

\section{Introduction}

Modern machine learning has been extraordinarily successful in domains where large volumes of labeled data are available~\cite{KrizhevskySH:12,Silver:16}.
Indeed, highly expressive models can generalize once they fit an appropriately large training set.
Unfortunately, many important domains are plagued by a scarcity of data or by expensive labels (or both).
One way to bridge this gap is by augmenting the given dataset with the help of a large family of innocuous distortions.
In many cases, the distortions correspond to the action of a group, meaning the ground truth exhibits known symmetries.
Augmenting the training set by applying the group action encourages the model to learn these symmetries.
While this approach has been successful~\cite{SimardSP:03,CiresanMGS:10,KrizhevskySH:12,ChenDL:20}, it is extremely inefficient to train a large, symmetry-agnostic model to find a highly symmetric function.
One wonders:
\begin{center}
\textit{Why not use a model that already accounts for known symmetries?}
\end{center}

This motivates \textit{invariant machine learning}~(e.g., \cite{Wood:96,ZhangL:04,GuEtal:18,VillarHSYB:21,BalanHS:22}), where the model is invariant to underlying symmetries in the data.
To illustrate, suppose an object is represented by a point $x$ in a set $V$, but there is a group $G$ acting on $V$ such that the same object is also represented by $gx\in V$ for every $g\in G$.
This ambiguity emerges, for example, when using a matrix to represent a point cloud or a graph, since the representation depends on the labeling of the points or vertices.
If we apply a $G$-invariant feature map $\Phi\colon V\to F$, then the learning task can be performed in the feature domain $F$ without having to worry about symmetries in the problem.
Furthermore, if $\Phi$ separates the $G$-orbits in $V$, then no information is lost by passing to the feature domain.
In practice, $V$ and $F$ tend to be vector spaces out of convenience, and $G$ is frequently a linear group.

While our interest in invariants stems from modern machine learning, maps like $\Phi$ have been studied since Cayley established \textit{invariant theory} in the nineteenth century~\cite{Cayley:45}.
Here, we take $V=\mathbb{C}^d$ and $G\leq\operatorname{GL}(V)$, and the maps of interest consist of the $G$-invariant polynomials $\mathbb{C}[V]^G$.
In 1890, Hilbert~\cite{Hilbert:90} proved that $\mathbb{C}[V]^G$ is finitely generated as a $\mathbb{C}$-algebra in the special case where $G$ is the image of a representation of $\operatorname{SL}(\mathbb{C}^k)$, meaning one may take the feature domain $F$ to be finite dimensional (one dimension for each generator).
Since $G$ is not a compact subset of $\mathbb{C}^{d\times d}$ in such cases, there may exist distinct $G$-orbits whose closures intersect, meaning no continuous $G$-invariant function can separate them; this subtlety plays an important role in Mumford's more general \textit{geometric invariant theory}~\cite{MumfordFK:94}.
In general, the generating set of $\mathbb{C}[V]^G$ is often extraordinarily large~\cite{Defresne:08}, making it impractical for machine learning applications.
To alleviate this issue, there has been some work to construct \textit{separating sets} of polynomials~\cite{DerksenK:15,Domokos:17,Domokos:arxiv}, i.e., sets that separate as well as $\mathbb{C}[V]^G$ does without necessarily generating all of $\mathbb{C}[V]^G$.
For every reductive group $G$, there exists a separating set of $2d+1$ invariant polynomials~\cite{Defresne:08,CahillCC:20}, but the complexity of evaluating these polynomials is still quite large.
Furthermore, these polynomials tend to have high degree, and so they are numerically unstable in practice.
In practice, one also desires a quantitative notion of separating so that distant orbits are not sent to nearby points in the feature space, and this behavior is not always afforded by a separating set of polynomials~\cite{CahillCC:20}.
Despite these shortcomings, polynomial invariants are popular in the data science literature due in part to their rich algebraic theory, e.g.,~\cite{BandeiraBKPWW:17,PerryWBRS:19,CahillCC:20,BendoryELS:22,BalanHS:22}.

In this paper, we focus on the case where $V$ is a real inner product space and $G$ is a group of linear isometries of $V$.
We introduce a family of non-polynomial invariants that we call \textit{max filters}.
In Section~\ref{sec.preliminaries}, we define max filters, we identify some basic properties, and we highlight a few familiar examples.
In Section~\ref{sec.separating}, we use ideas from~\cite{DymG:22} to establish that $2d$ generic max filters separate all $G$-orbits when $G$ is finite (see Corollary~\ref{cor.2d templates suffice for finite groups}), and then we describe various settings in which max filtering is computationally efficient.
In Section~\ref{sec.lipschitz}, we show that when $G$ is finite, a sufficiently large random max filter bank is bilipschitz with high probability; see Theorem~\ref{thm.main result}.
This is the first known construction of invariant maps for a general class of groups that enjoy a lower Lipschitz bound, meaning they separate orbits in a quantitative sense.
In the same section, we later show that when $V=L^2(\mathbb{R}^d)$ and $G$ is the group of translations, certain max filters exhibit stability to diffeomorphic distortion akin to what Mallat established for his scattering transform in~\cite{Mallat:12}; see Theorem~\ref{thm.mallat bound}.
In Section~\ref{sec.classification}, we explain how to select max filters for classification in a couple of different settings, we determine the subgradient of max filters to enable training, and we characterize how random max filters behave for the symmetric group.
In Section~\ref{sec.numerics}, we use max filtering to process real-world datasets.
Specifically, we visualize the shape space of voting districts, we use electrocardiogram data to classify whether patients had a heart attack, and we classify a multitude of textures.
Surprisingly, we find that even in cases where the data do not appear to exhibit symmetries in a group $G$, max filtering with respect to $G$ can still reveal salient features.
We conclude in Section~\ref{sec.discussion} with a discussion of opportunities for follow-on work.

\section{Preliminaries}
\label{sec.preliminaries}

Given a real inner product space $V$ and a group $G$ of linear isometries of $V$, consider the quotient space $V/G$ consisting of the $G$-orbits $[x]:=G\cdot x$ for $x\in V$.
This quotient space is equipped with the metric
\[
d([x],[y])
:=\inf_{p\in[x],q\in[y]} \|p-q\|.
\]
(Indeed, $d$ satisfies the triangle inequality since $G$ is a group of isometries of $V$.)
This paper is concerned with the following function.

\begin{definition}
The \textbf{max filtering map} $\llangle \cdot,\cdot\rrangle\colon V/G\times V/G\to \mathbb{R}$ is defined by
\[
\llangle [x],[y]\rrangle
:=\sup_{p\in[x],q\in[y]}\langle p,q\rangle.
\]
Sometimes, ``max filtering map'' refers to the related function $\llangle [\cdot],[\cdot]\rrangle\colon V\times V\to \mathbb{R}$ instead. The intended domain should be clear from context.
\end{definition}

We note that since $G$ consists of linear isometries, it is closed under adjoints, and so the max filtering map can be alternatively expressed as
\[
\llangle [x],[y]\rrangle
=\sup_{g\in G}\langle gx,y\rangle
=\sup_{g\in G}\langle x,gy\rangle.
\]
Furthermore, if $G$ is topologically closed (e.g., finite), then the supremum can be replaced with a maximum:
\[
\llangle [x],[y]\rrangle
=\max_{g\in G}\langle gx,y\rangle
=\max_{g\in G}\langle x,gy\rangle.
\]
The max filtering map satisfies several important properties, summarized below.

\begin{lemma}
\label{lem.max filter product properties}
Suppose $x,y,z\in V$.
Then each of the following holds:
\begin{itemize}
\item[(a)]
$\llangle [x],[x]\rrangle=\|x\|^2$.
\item[(b)]
$\llangle [x],[y]\rrangle=\llangle [y],[x]\rrangle$.
\item[(c)]
$\llangle [x],[ry]\rrangle=r\llangle [x],[y]\rrangle$ for every $r\geq0$.
\item[(d)]
$\llangle [x],[\cdot]\rrangle\colon V\to\mathbb{R}$ is convex.
\item[(e)]
$\llangle [x],[y+z]\rrangle\leq\llangle [x],[y]\rrangle+\llangle [x],[z]\rrangle$.
\item[(f)]
$d([x],[y])^2=\|x\|^2-2\llangle [x],[y]\rrangle+\|y\|^2$.
\item[(g)]
$\llangle [x],\cdot\rrangle\colon V/G\to\mathbb{R}$ is $\|x\|$-Lipschitz.
\end{itemize}
\end{lemma}

\begin{proof}
First, (a) and (b) are immediate, as is the $r=0$ case of (c).
For $r>0$, observe that $q\in[ry]$ precisely when $q':=r^{-1}q\in[y]$, since each member of $G$ is a linear isometry.
Thus,
\[
\llangle [x],[ry]\rrangle
=\sup_{p\in[x],q\in[ry]}\langle p,q\rangle
=\sup_{p\in[x],q'\in[y]}\langle p,rq'\rangle
=r\cdot\sup_{p\in[x],q'\in[y]}\langle p,q'\rangle
=r\llangle [x],[y]\rrangle.
\]
Next, (d) follows from the identity $\llangle [x],[y]\rrangle=\sup_{p\in[x]}\langle p,y\rangle$, which expresses $\llangle [x],[\cdot]\rrangle$ as a pointwise supremum of convex functions.
For (e), we apply (d) and (c):
\[
\llangle [x],[y+z]\rrangle
=\llangle [x],[\tfrac{1}{2}\cdot2y+\tfrac{1}{2}\cdot2z]\rrangle
\leq\tfrac{1}{2}\llangle [x],[2y]\rrangle+\tfrac{1}{2}\llangle [x],[2z]\rrangle
=\llangle [x],[y]\rrangle+\llangle [x],[z]\rrangle.
\]
Next, (f) is immediate.
For (g), select any $y,z\in V$.
We may assume $\llangle [x],[y]\rrangle\geq \llangle [x],[z]\rrangle$ without loss of generality.
Select any $p\in[y]$, $q\in[z]$, and $\epsilon>0$, and take $g\in G$ such that $\langle x,gp\rangle>\llangle [x],[y]\rrangle-\epsilon$.
Then
\begin{align*}
|\llangle [x],[y]\rrangle-\llangle [x],[z]\rrangle|
&=\llangle [x],[y]\rrangle-\llangle [x],[z]\rrangle\\
&< \langle x,gp\rangle+\epsilon-\langle x,gq\rangle
=\langle x,g(p-q)\rangle+\epsilon
\leq\|x\|\|p-q\|+\epsilon.
\end{align*}
Since $p$, $q$, and $\epsilon$ were arbitrary, the result follows.
\end{proof}

\begin{definition}
Given a \textbf{template} $z\in V$, we refer to $\llangle [z],\cdot\rrangle\colon V/G\to\mathbb{R}$ as the corresponding \textbf{max filter}.
Given a (possibly infinite) sequence $\{z_i\}_{i\in I}$ of templates in $V$, the corresponding \textbf{max filter bank} is $\Phi\colon V/G\to \mathbb{R}^I$ defined by $\Phi([x]):=\{\llangle [z_i],[x]\rrangle\}_{i\in I}$.
\end{definition}

In what follows, we identify a few familiar examples of max filters.

\begin{example}[Norms]
There are several norms that can be thought of as a max filter with some template.
For example, consider $V=\mathbb{R}^n$.
Then taking $G=\operatorname{O}(n)$ and any unit-norm template $z$ gives
\[
\llangle [z],[x]\rrangle
=\max_{g\in\operatorname{O}(n)}\langle gz,x\rangle
=\|x\|.
\]
Similarly, the infinity norm is obtained by taking $G$ to be the group of signed permutation matrices and $z$ to be a standard basis element, while the $1$-norm comes from taking $G$ to be the group of diagonal orthogonal matrices and $z$ to be the all-ones vector.
We can also recover various matrix norms when $V=\mathbb{R}^{m\times n}$.
For example, taking $G\cong\operatorname{O}(m)\times\operatorname{O}(n)$ to be the group of linear operators of the form $X\mapsto Q_1XQ_2^{-1}$ for $Q_1\in\operatorname{O}(m)$ and $Q_2\in\operatorname{O}(n)$, then max filtering with any rank-$1$ matrix of unit Frobenius norm gives the spectral norm.
\end{example}

\begin{example}[Power spectrum]
Consider the case where $V=L^2(\mathbb{R}/\mathbb{Z})$ and $G$ is the group of circular translation operators $T_a$ defined by $T_ag(t):=g(t-a)$ for $a\in\mathbb{R}$.
(Here and throughout, functions in $L^2$ will be real valued by default.)
Given a template $z_k$ of the form $z_k(t):=\cos(2\pi kt)$ for some $k\in\mathbb{N}$, it holds that
\begin{align*}
\llangle [z_k],[f]\rrangle
=\max_{a\in [0,1)}\langle T_az_k,f\rangle
&=\max_{a\in [0,1)}\int_0^1\cos(2\pi k(t-a))f(t)dt\\
&=\max_{a\in [0,1)}\operatorname{Re}\bigg(e^{2\pi ika}\int_0^1 f(t)e^{-2\pi ikt}dt\bigg)
=|\hat{f}(k)|.
\end{align*}
A similar choice of templates recovers the power spectrum over finite abelian groups.
\end{example}

\begin{example}[Unitary groups]
\label{ex.complex}
While we generally assume $V$ is a real inner produce space, our theory also applies in the complex setting.
For example, consider the case where $V=\mathbb{C}^n$ and $G\leq\operatorname{U}(n)$.
Then $V$ is a $2n$-dimensional real inner product space with
\[
\langle x,y\rangle
:=\operatorname{Re}(x^*y),
\]
where $x^*$ denotes the conjugate transpose of $x$.
As such, $\operatorname{U}(n)\leq\operatorname{O}(V)$ since $g\in\operatorname{U}(n)$ implies
\[
\langle gx,gy\rangle
=\operatorname{Re}((gx)^*(gy))
=\operatorname{Re}(x^*g^*gy)
=\operatorname{Re}(x^*y)
=\langle x,y\rangle.
\]
Thus, $G\leq\operatorname{O}(V)$.
\end{example}

\begin{example}[Phase retrieval]
\label{ex.phase retrieval}
Suppose $V=\mathbb{C}^r$ and $G=\{c\cdot\operatorname{id}:|c|=1\}\leq\operatorname{U}(r)$.
Then
\[
\llangle [z],[x]\rrangle
=\max_{|c|=1}\operatorname{Re}((cz)^*x)
=|z^*x|.
\]
The max filter bank corresponding to $\{z_i\}_{i=1}^n$ in $V$ is given by $\Phi([x])=\{|z_i^*x|\}_{i=1}^n$.
The inverse problem of recovering $[x]$ from $\Phi([x])$ is known as \textit{complex phase retrieval}~\cite{BalanCE:06,BandeiraCMN:14,ConcaEHV:15,Vinzant:15}, and over the last decade, several algorithms were developed to solve this inverse problem~\cite{CandesSV:13,DemanetH:14,CandesESV:15,WaldspurgerdAM:15,CandesLS:15,ChenC:17}.
In the related setting where $V=\mathbb{R}^d$ and $G=\{\pm\operatorname{id}\}\leq\operatorname{O}(d)$, the analogous inverse problem is known as \textit{real phase retrieval}~\cite{BalanCE:06}.
\end{example}

\begin{example}[Matched filtering]
In classical radar, the primary task is to locate a target.
Here, a transmitter emits a pulse $p\in L^2(\mathbb{R})$, which then bounces off the target and is received at the transmitter's location with a known direction of arrival.
The return signal $q$ is a noisy version of $T_ap$ for some $a>0$, where $T_a$ denotes the translation-by-$a$ operator defined by $T_af(t):=f(t-a)$.
Considering the transmitter-to-target distance is $a/2$ times the speed of light, the objective is to estimate $a$, which can be accomplished with \textit{matched filtering}: simply find $a$ for which $\langle T_ap,q\rangle$ is largest.
This is essentially a max filter with $V=L^2(\mathbb{R})$ and $G$ being the group of translation operators, though for this estimation problem, the object of interest is the maximizer $a$, not the maximum value $\llangle [p],[q]\rrangle$.
Meanwhile, the maximum value is used for the detection problem of distinguishing noise from noisy versions of translates of $p$.
(This accounts for half of the etymology of \textit{max filtering}.)
\end{example}

\begin{example}[Max pooling]
In a convolutional neural network, it is common for a convolutional layer to be followed by a max pooling layer.
Here, the convolutional layer convolves the input image with several localized templates, and then the max pooling layer downsamples each of the resulting convolutions by partitioning the scene into patches and recording the maximum value in each patch.
In the extreme case where the max pooling layer takes the entire scene to be a single patch to maximize over, these layers implement a max filter bank in which $V$ is the image space and $G$ is the group of translation operators.
(This accounts for the other half of the etymology of \textit{max filtering}.)
\end{example}

\section{The complexity of separating orbits}
\label{sec.separating}

For practical reasons, we are interested in orbit-separating invariants $\Phi\colon V\to\mathbb{R}^n$ of \textit{low complexity}, which we take to mean two different things simultaneously:
\begin{itemize}
\item[(i)]
$n$ is small (i.e., the map has low \textit{sample complexity}), and
\item[(ii)]
one may evaluate $\Phi$ efficiently (i.e., the map has low \textit{computational complexity}). 
\end{itemize}
While these notions of complexity are related, they impact the learning task in different ways.
In what follows, we study both notions of complexity in the context of max filtering.

\subsection{Generic templates separate orbits}

In this subsection, we focus on the case in which $V=\mathbb{R}^d$ and $G\leq\operatorname{O}(d)$ is \textit{semialgebraic}, which we will define shortly.
Every polynomial function $p\colon\mathbb{R}^n\to\mathbb{R}$ determines a basic semialgebraic set
\[
\{x\in\mathbb{R}^n:p(x)\geq0\}.
\]
By closing under finite unions, finite intersections, and complementation, the basic semialgebraic sets of $\mathbb{R}^n$ generate an algebra of sets known as the \textbf{semialgebraic sets} in $\mathbb{R}^n$.
A \textbf{semialgebraic subgroup} $G\leq\operatorname{GL}(d)$ is a subgroup of $\operatorname{GL}(d)$ that is also a semialgebraic set in $\mathbb{R}^{d\times d}$, e.g., $\operatorname{O}(d)$.
A \textbf{semialgebraic function} is a function $f\colon\mathbb{R}^s\to\mathbb{R}^t$ for which the graph $\{(x,f(x)):x\in\mathbb{R}^s\}$ is a semialgebraic set in $\mathbb{R}^{s+t}$.

\begin{lemma}
\label{lem.max filtering is semialgebraic}
For every semialgebraic subgroup $G\leq\operatorname{O}(d)$, the corresponding max filtering map $\llangle [\cdot],[\cdot]\rrangle\colon\mathbb{R}^d\times\mathbb{R}^d\to\mathbb{R}$ is semialgebraic.
\end{lemma}

\begin{proof}
The graph of the max filtering map can be expressed in first-order logic:
\[
\{(x,y,t)\in\mathbb{R}^d\times\mathbb{R}^d\times\mathbb{R}:
(\forall g\in G,t\geq\langle x,gy\rangle)
\wedge (\forall \epsilon\in\mathbb{R}, \exists g\in G,\epsilon>0\Rightarrow t-\epsilon<\langle x,gy\rangle)\}.
\]
(To be precise, one should replace our quantifiers over $G$ with the polynomial conditions that define the semialgebraic set $G$ to obtain a condition in first-order logic.)
It follows from Proposition~2.2.4 in~\cite{BochnakCR:13} that the graph is semialgebraic.
\end{proof}

Every semialgebraic set $A$ can be decomposed as a disjoint union $A=\bigcup_i A_i$ of finitely many semialgebraic sets $A_i$, each of which is homeomorphic to an open hypercube $(0,1)^{d_i}$ (where $(0,1)^0$ is a point).
The \textbf{dimension} of $A$ can be defined in terms of this decomposition as $\operatorname{dim}(A):=\max_i d_i$.
(It does not depend on the decomposition.)

\begin{definition}
Given a semialgebraic subgroup $G\leq\operatorname{O}(d)$, we say the corresponding max filtering map $\llangle [\cdot],[\cdot]\rrangle\colon\mathbb{R}^d\times\mathbb{R}^d\to\mathbb{R}$ is $k$-\textbf{strongly separating} if for every $x,y\in\mathbb{R}^d$ with $[x]\neq[y]$, it holds that
\[
\operatorname{dim}\big\{z\in\mathbb{R}^d:\llangle [z],[x]\rrangle=\llangle [z],[y]\rrangle\big\}
\leq d-k.
\]
\end{definition}

As an example, consider the case where $G=\operatorname{O}(d)$.
Then $\llangle [z],[x]\rrangle=\llangle [z],[y]\rrangle$ holds precisely when $\|z\|\|x\|=\|z\|\|y\|$, i.e., $z=0$ or $[x]=[y]$.
Thus, the max filtering map is $d$-strongly separating in this case.

\begin{theorem}
\label{thm.dym-gortler}
Consider any semialgebraic subgroup $G\leq\operatorname{O}(d)$ with $k$-strongly separating max filtering map $\llangle [\cdot],[\cdot]\rrangle\colon\mathbb{R}^d\times\mathbb{R}^d\to\mathbb{R}$ for some $k\in\mathbb{N}$.
For generic $z_1,\ldots,z_n\in\mathbb{R}^d$, the max filter bank $x\mapsto\{\llangle[z_i],[x]\rrangle\}_{i=1}^n$ separates $G$-orbits in $\mathbb{R}^d$ provided $n\geq 2d/k$.
\end{theorem}

In the $d$-strongly separating case where $G=\operatorname{O}(d)$, Theorem~\ref{thm.dym-gortler} implies that $n=2$ generic templates suffice to separate orbits.
(Of course, any single nonzero template suffices in this case.)
Theorem~\ref{thm.dym-gortler} is an improvement to Theorem~1.9 in~\cite{DymG:22}, which gives the condition $n\geq 2d+1$.
We obtain the improvement $n\geq 2d/k$ by leveraging a more detailed notion of \textit{strongly separating}, as well as the positive homogeneity of max filtering.
The proof makes use of a \textit{lift-and-project} technique that first appeared in~\cite{BalanCE:06} and was subsequently applied in~\cite{ConcaEHV:15,WangX:19,CahillCC:20,RongWX:21}.

\begin{proof}[Proof of Theorem~\ref{thm.dym-gortler}]
Fix $n\geq 2d$, and let $\mathcal{Z}\subseteq(\mathbb{R}^d)^n$ denote the set of $\{z_i\}_{i=1}^n$ for which the max filter bank $x\mapsto\{\llangle[z_i],[x]\rrangle\}_{i=1}^n$ fails to separate $G$-orbits in $\mathbb{R}^d$.
We will show that $\mathcal{Z}$ is semialgebraic with dimension $\leq dn-1$, from which the result follows.
To do so, observe that $\{z_i\}_{i=1}^n\in\mathcal{Z}$ precisely when there exists a \textit{witness}, namely, $(x,y)\in\mathbb{R}^d\times\mathbb{R}^d$ with $[x]\neq[y]$ such that $\llangle[z_i],[x]\rrangle=\llangle[z_i],[y]\rrangle$ for every $i\in\{1,\ldots,n\}$.
In fact, we may assume that the witness $(x,y)$ satisfies $\|x\|^2+\|y\|^2=1$ without loss of generality since the set of witnesses for $\{z_i\}_{i=1}^n$ avoids $(0,0)$ and is closed under positive scalar multiplication by Lemma~\ref{lem.max filter product properties}(c).
This suggests the following lift of $\mathcal{Z}$:
\begin{align*}
\mathcal{L}
:=\Big\{~(\{z_i\}_{i=1}^n,(x,y))\in(\mathbb{R}^d)^n\times(\mathbb{R}^d)^2~:
&~ [x]\neq[y], ~ \|x\|^2+\|y\|^2=1,\\
&~ \llangle[z_i],[x]\rrangle=\llangle[z_i],[y]\rrangle ~~\forall i\in\{1,\ldots,n\}~\Big\}.
\end{align*}
Since $G$ is semialgebraic, we have that $[x]\neq[y]$ is a semialgebraic condition.
Furthermore, $\llangle[z_i],[x]\rrangle=\llangle[z_i],[y]\rrangle$ is a semialgebraic condition for each $i$ by Lemma~\ref{lem.max filtering is semialgebraic}.
It follows that $\mathcal{L}$ is semialgebraic.
Next, we define the projection maps $\pi_1\colon(\{z_i\}_{i=1}^n,(x,y))\mapsto \{z_i\}_{i=1}^n$ and $\pi_2\colon(\{z_i\}_{i=1}^n,(x,y))\mapsto (x,y)$.
Then $\mathcal{Z}=\pi_1(\mathcal{L})$ is semialgebraic by Tarski--Seidenberg (Proposition~2.2.1 in~\cite{BochnakCR:13}).
To bound the dimension of $\mathcal{Z}$, we first observe that
\[
\pi_2^{-1}(x,y)
=\big\{z\in\mathbb{R}^d:\llangle [z],[x]\rrangle=\llangle [z],[y]\rrangle\big\}^n \times\{(x,y)\},
\]
and so $\operatorname{dim}(\pi_2^{-1}(x,y)) \leq n(d-k)$, since the max filtering map is $k$-strongly separating by assumption.
We use the fact that $\pi_2(\mathcal{L})$ is contained in the unit sphere in $(\mathbb{R}^d)^2$ together with Lemma~1.10 in~\cite{DymG:22} to obtain
\[
\operatorname{dim}(\mathcal{Z})
\leq\operatorname{dim}(\mathcal{L})
\leq\operatorname{dim}(\pi_2(\mathcal{L}))+\max_{x,y\in\mathbb{R}^d}\operatorname{dim}(\pi_2^{-1}(x,y))
\leq (2d-1)+n(d-k)
\leq dn-1,
\]
where the last step is equivalent to the assumption $n\geq 2d/k$.
\end{proof}

\begin{corollary}
\label{cor.2d templates suffice for finite groups}
Consider any finite subgroup $G\leq\operatorname{O}(d)$.
For generic $z_1,\ldots,z_n\in\mathbb{R}^d$, the max filter bank $x\mapsto\{\llangle[z_i],[x]\rrangle\}_{i=1}^n$ separates $G$-orbits in $\mathbb{R}^d$ provided $n\geq 2d$.
\end{corollary}

Corollary~\ref{cor.2d templates suffice for finite groups} follows immediately from Theorem~\ref{thm.dym-gortler} and the following lemma:

\begin{lemma}
For every finite subgroup $G\leq\operatorname{O}(d)$, the corresponding max filtering map $\llangle [\cdot],[\cdot]\rrangle\colon\mathbb{R}^d\times\mathbb{R}^d\to\mathbb{R}$ is $1$-strongly separating.
\end{lemma}

\begin{proof}
Consider any $x,y\in\mathbb{R}^d$ with $[x]\neq[y]$.
Then $\llangle [z],[x]\rrangle=\llangle [z],[y]\rrangle$ only if there exists $g\in G$ such that $\langle z,x\rangle=\langle z,gy\rangle$, i.e., $z\in\operatorname{span}\{x-gy\}^\perp$.
Thus,
\[
\big\{z\in\mathbb{R}^d:\llangle [z],[x]\rrangle=\llangle [z],[y]\rrangle\big\}
\subseteq\bigcup_{g\in G}\operatorname{span}\{x-gy\}^\perp.
\]
Since the max filtering map is semialgebraic by Lemma~\ref{lem.max filtering is semialgebraic}, it follows that the left-hand set is also semialgebraic.
Since $G$ is finite, the right-hand set is semialgebraic with dimension $d-1$, and the result follows.
\end{proof}

We would like to know if a version of Corollary~\ref{cor.2d templates suffice for finite groups} holds for all semialgebraic groups, but we do not have a proof of strongly separating for infinite groups in general.
This motivates the following problem:

\begin{problem}\
\label{prob.semialgebraic}
\begin{itemize}
\item[(a)]
For which semialgebraic groups is the max filtering map $k$-strongly separating?
\item[(b)]
How many templates are needed to separate orbits for a given group?
\end{itemize}
\end{problem}

We identify a couple of interesting instances of Problem~\ref{prob.semialgebraic}.
First, we consider the case of complex phase retrieval (as in Example~\ref{ex.phase retrieval}), where $V=\mathbb{C}^r$ and $G$ is the center of $\operatorname{U}(r)$.
It is known that $n=4r-4=2\operatorname{dim}(V)-4$ generic templates separate orbits for every $r$, and this is the optimal threshold for infinitely many $r$~\cite{ConcaEHV:15}, but there also exist $11$ templates in $\mathbb{C}^4$ that separate orbits, for example~\cite{Vinzant:15}.

As another example, consider the case where $V=\mathbb{R}^d$ and $G\cong S_d$ is the group of $d\times d$ permutation matrices.
Then Corollary~\ref{cor.2d templates suffice for finite groups} gives that $2d$ generic templates separate orbits.
However, it is straightforward to see that the templates $z_j:=\sum_{i=1}^j e_i$ for $j\in\{1,\ldots,d\}$ also separate orbits, where $e_i$ denotes the $i$th standard basis element.
Indeed, take $\operatorname{sort}(x)$ to have weakly decreasing entries.
Then the first entry equals $\llangle [z_1],[x]\rrangle$, while for each $j>1$, the $j$th entry equals $\llangle [z_j],[x]\rrangle-\llangle [z_{j-1}],[x]\rrangle$.
As such, this max filter bank determines $\operatorname{sort}(x)$, which is a separating invariant of $V/G$.
Considering Theorem~\ref{thm.dym-gortler}, one might suspect that the max filtering map is $2$-strongly separating in this case, but this is not so.
Indeed, the cone $C$ of sorted vectors in $\mathbb{R}^d$ has  dimension~$d$, and so there exists a subspace $H$ of co-dimension~$1$ that intersects the interior of $C$.
Select any $x$ in the interior of $C$ and any unit vector $v\in H^\perp$, and then take $y=x+\epsilon v$ for $\epsilon>0$ sufficiently small so that $y\in C$.
Then 
\[
\big\{z\in\mathbb{R}^d:\llangle [z],[x]\rrangle=\llangle [z],[y]\rrangle\big\}
\supseteq H\cap C,
\]
which is a semialgebraic set of dimension $d-1$, and so the claim follows.

\subsection{Low-complexity max filtering}

In this subsection, we focus on the case in which $V\cong\mathbb{R}^d$.
Naively, one may compute the max filtering map $(x,y)\mapsto\llangle [x],[y]\rrangle$ over a finite group $G\leq\operatorname{O}(d)$ of order $m$ by computing $\langle x,gy\rangle$ for every $g\in G$ and then returning the maximum.
This approach costs $O(md)$ operations.
Of course, this is not possible when $G$ is infinite, and it is prohibitive when $G$ is finite but large.
Interestingly, many of the groups that we encounter in practice admit a faster implementation.
In particular, for many quotient spaces $V/G$, the quotient metric $d\colon V/G\times V/G\to\mathbb{R}$ is easy to compute, and by Lemma~\ref{lem.max filter product properties}(f), the max filtering map is equally easy to compute in these cases:
\[
\llangle [x],[y]\rrangle
=\frac{1}{2}\Big(d([x],[y])^2-\|x\|^2-\|y\|^2\Big).
\]
In this subsection, we highlight a few examples of such quotient spaces before considering the harder setting of graphs.

\subsubsection{Point clouds}

Consider $V=\mathbb{R}^{k\times n}$ with Frobenius inner product.
We can represent a point cloud of $n$ points in $\mathbb{R}^k$ as a member of $V$ by arbitrarily labeling the points with column indices.
In this setting, we identify members of $V$ that reside in a common $G$-orbit with $G\cong S_n$ permuting the columns.
The resulting quotient metric is known as the \textit{$2$-Wasserstein distance}:
\[
d([X],[Y])
=\min_{P\in\Pi(n)}\|X-YP\|_F,
\]
where $\Pi(n)$ denotes the set of $n\times n$ permutation matrices and $\|\cdot\|_F$ denotes the Frobenius norm.
The corresponding max filtering map is then given by
\[
\llangle [X],[Y]\rrangle
=\max_{P\in\Pi(n)}\langle X,YP\rangle
=\max_{P\in\Pi(n)}\operatorname{tr}(X^\top YP)
=\max_{S\in\operatorname{conv}\Pi(n)}\operatorname{tr}(X^\top YS),
\]
where $\operatorname{conv}\Pi(n)$ denotes the convex hull of $\Pi(n)$, namely, the doubly stochastic matrices.
By this formulation, the max filtering map can be computed in polynomial time by linear programming.
In the special case where $k=1$, the max filtering map has an even faster implementation:
\[
\llangle [X],[Y]\rrangle
=\langle \operatorname{sort}(X), \operatorname{sort}(Y) \rangle,
\]
which can be computed in linearithmic time.

\subsubsection{Circular translations}

Consider the case where $V\cong\mathbb{R}^n$ is the space of vectors with entries indexed by the cyclic group $C_n:=\mathbb{Z}/n\mathbb{Z}$, and $G\cong C_n$ is the group of circular translations $T_a$ defined by $T_af(x):=f(x-a)$.
Then the max filtering map is given by
\[
\llangle [f],[g]\rrangle
=\max_{a\in C_n}\langle f,T_ag\rangle
=\max_{a\in C_n}\sum_{x\in C_n}f(x)g(x-a)
=\max_{a\in C_n}(f\star Rg)(a),
\]
where $\star$ denotes the circular convolution and $R$ denotes the reversal operator.
Thus, the max filtering map can be computed in linearithmic time with the help of the fast Fourier transform.

\subsubsection{Shape analysis}

In geometric morphometics~\cite{MitteroeckerG:09}, it is common for data to take the form of a sequence of $n$ landmarks in $\mathbb{R}^k$ (where $k$ is typically $2$ or $3$) with a global rotation ambiguity.
This corresponds to taking $V=\mathbb{R}^{k\times n}$ and $G\cong\operatorname{O}(k)$ acting on the left, and so the max filtering map is given by
\[
\llangle [X],[Y]\rrangle
=\max_{R\in\operatorname{O}(k)}\langle X,RY\rangle
=\max_{R\in\operatorname{O}(k)}\operatorname{tr}(RYX^\top)
=\|YX^\top\|_*,
\]
where $\|\cdot\|_*$ denotes the nuclear norm.
As such, the max filtering map can be computed in polynomial time with the aid of the singular value decomposition.

\subsubsection{Separation hierarchy for weighted graphs}

Here, we focus on the case in which $V$ is the vector space of real symmetric $n\times n$ matrices with zero diagonal and $G\cong S_n$ is the group of linear isometries of the form $A\mapsto PAP^{-1}$, where $P$ is a permutation matrix.
We think of $V/G$ as the space of weighted graphs on $n$ vertices (up to isomorphism).
One popular approach for separating graphs uses \textit{message-passing graph neural networks}, but the separation power of such networks is limited by the so-called \textit{Weisfeler--Lehman test}~\cite{WeisfeilerA:68,XuHLJ:19,MorrisRFHLRG:19}.
For example, message-passing graph neural networks fail to distinguish $C_3\cup C_3$ from $C_6$.
See~\cite{Sato:20,HuangV:21} for surveys of this rapidly growing literature.

As an alternative, we consider a max filtering approach.
Given two adjacency matrices $A_1$ and $A_2$, Lemma~\ref{lem.max filter product properties}(f) implies that the corresponding graphs are isomorphic if and only if
\[
\|A_1\|_F^2=\llangle [A_1],[A_2]\rrangle=\|A_2\|_F^2.
\]
As such, max filtering is graph isomorphism--hard in this setting.
Interestingly, there exist $A\in V$ for which the map $X\mapsto\llangle [A],[X]\rrangle$ can be computed in linearithmic time, and furthermore, these easy-to-compute max filters help with separating orbits.
To see this, we follow~\cite{AlonYZ:95}, which uses \textit{color coding} to facilitate computation by dynamic programming.

\begin{definition}
A tuple $\{f_i\}_{i=1}^N$ with $f_i\colon[n]\to[k]$ for each $i\in [N]$ is an \textbf{$(n,k)$-color coding} if $n\geq k$ and for every $S\subseteq[n]$ of cardinality $k$, there exists $i\in [N]$ such that $f_i(S)=[k]$.
\end{definition}

\begin{lemma}
\label{lem.color coding size}
Given $n,k\in\mathbb{N}$ with $n\geq k$, there exists an $(n,k)$-color coding of size $\lceil ke^k\log n\rceil$.
\end{lemma}

\begin{proof}
We show that $N$ random colorings form a color coding with positive probability.
For each $i\in [N]$ and $S\in\binom{[n]}{k}$, we have $\mathbb{P}\{f_i(S)=[k]\}=k!/k^k$, and so the union bound gives
\[
\mathbb{P}\Big\{\{f_i\}_{i=1}^N\text{ is not $(n,k)$-color coding}\Big\}
\leq\sum_{S\in\binom{[n]}{k}} \prod_{i=1}^N\mathbb{P}\{f_i(S)\neq [k]\}
=\binom{n}{k}\bigg(1-\frac{k!}{k^k}\bigg)^{N}.
\]
It suffices to select $N$ so that the right-hand side is strictly smaller than $1$.
The result follows by applying the bounds $\binom{n}{k}\leq n^k$, $k!\geq (k/e)^k$, and $(1-1/t)^t<1/e$ for $t>1$:
\[
\tbinom{n}{k}(1-\tfrac{k!}{k^k})^{N}
\leq n^k(1-e^{-k})^{e^k e^{-k}N}
< e^{k\log n-e^{-k}N},
\]
which is at most $1$ when $N\geq ke^k\log n$.
\end{proof}

\begin{algorithm}[t]
\SetAlgoLined
\KwData{Weighted tree with vertex set $[k]$ (in post-order traversal order) and with adjacency matrix $A\in\mathbb{R}^{k\times k}$, and $(n,k)$-color coding $\{f_i\}_{i=1}^N$}
\KwIn{Weighted graph with adjacency matrix $B\in\mathbb{R}^{n\times n}$}
\KwResult{Max filter $\llangle [\tilde{A}],[B]\rrangle$, where $\tilde{A}:=[\begin{smallmatrix} A&0\\0&0\end{smallmatrix}]\in\mathbb{R}^{n\times n}$}
\For{$i\in[N]$ and $\pi\in S_k$}{
Initialize $\ell(v)\leftarrow0$ for all $v\in [n]$\\
\For{$u\in[k-1]$}{
Take the unique vertex $u' \in \{u+1,\ldots,k\}$ adjacent to the leaf $u$ in $H-[u-1]$\\
\For{$v'\in (\pi\circ f_i)^{-1}(u')$}{
$\ell(v')\leftarrow\ell(v')+\max_{v\in (\pi\circ f_i)^{-1}(u)}(\ell(v)+A_{u,u'}B_{v,v'})$\\
}
}
Put $s(i,\pi):=\max_{v\in (\pi\circ f_i)^{-1}(k)}\ell(v)$\\
}
Output $\max_{i\in[N],\pi\in S_k}s(i,\pi)$\\
\caption{Max filtering with a weighted tree template by color coding
 \label{alg.color_coded_max_filter}}
\end{algorithm}

Algorithm~\ref{alg.color_coded_max_filter} computes the max filter with a small weighted tree using a color coding and dynamic programming.
Lemma~\ref{lem.color coding size} implies that Algorithm~\ref{alg.color_coded_max_filter} has runtime $e^{O(k\log k)}n^2\log n$, which is linearithmic in the size of the data when $k$ is fixed.
Notice that max filtering with the path on $k=4$ vertices already separates the graphs $C_3\cup C_3$ and $C_6$.
Furthermore, using techniques from~\cite{AlonYZ:95}, one can modify Algorithm~\ref{alg.color_coded_max_filter} to max filter with \textit{any} template graph $H$ on $k$ vertices, though the runtime becomes $e^{O(k\log k)}n^{t+1}\log n$, where $t$ is the \textit{treewidth} of $H$.
Letting $\mathcal{H}(k,t)$ denote the set of weighted graphs on at most $k$ vertices with treewidth at most $t$, we have the following hierarchy:
\begin{equation}
\label{eq.hierarchy}
\begin{array}{ccccccccc}
\mathcal{H}(n,1)&\subseteq&\cdots&\subseteq&\mathcal{H}(n,n-1)\\
\vsubseteq\\
\vdots&&\iddots\\
\vsubseteq\\
\mathcal{H}(2,1)
\end{array}
\end{equation}
Corollary~\ref{cor.2d templates suffice for finite groups} gives that $n(n-1)$ generic templates from $\mathcal{H}(n,n-1)$ separate all isomorphism classes of weighted graphs on $n$ vertices.
It would be interesting to study the separation power of templates of logarithmic order and bounded treewidth.

\section{Stability of max filtering}
\label{sec.lipschitz}

\subsection{Bilipschitz max filter banks}

Upper and lower Lipschitz bounds are used to quantify the stability of a mapping between metric spaces, but it is generally difficult to estimate such bounds; see~\cite{BandeiraCMN:14,BalanW:15,CahillCD:16,IwenMP:19,BalanD:21} for examples from phase retrieval and~\cite{ZouBS:19,CahillCC:20,CahillCC:arxiv,BalanHS:22} for other examples.
In this subsection, we prove the following:

\begin{theorem}
\label{thm.main result}
Fix a finite group $G\leq\operatorname{O}(d)$ of order $m$ and select
\[
n\geq 12m^2d\log(\tfrac{2}{\delta}+1),
\qquad
\delta
:=(\tfrac{\pi}{128m^4}\cdot\tfrac{1}{2d+3\log(4m^2)})^{1/2}.
\]
Draw independent random vectors $z_1,\ldots,z_n\sim\mathsf{Unif}(S^{d-1})$.
With probability $\geq1-e^{-n/(12m^2)}$, it holds that the max filter bank $\Phi\colon\mathbb{R}^d/G\to\mathbb{R}^n$ with templates $\{z_i\}_{i=1}^n$ has lower Lipschitz bound $\delta$ and upper Lipschitz bound $n^{1/2}$.
\end{theorem}

This result distinguishes max filtering from separating polynomial invariants, which do not necessarily enjoy upper or lower Lipschitz bounds~\cite{CahillCC:20}.
In Theorem~\ref{thm.main result}, we may take the embedding dimension to be $n=\Theta^*(m^2d)$ with bilipschitz bounds $\Theta^*(\frac{1}{m^2d^{1/2}})$ and $\Theta^*(md^{1/2})$, where $\Theta^*(\cdot)$ suppresses logarithmic factors.
For comparison, we consider a couple of cases that have already been studied in the literature.
First, the case where $G=\{\pm\operatorname{id}\}$ reduces to the setting of \textit{real phase retrieval} (as in Example~\ref{ex.phase retrieval}), where it is known that there exist $n=\Theta(d)$ templates that deliver lower- and upper-Lipschitz bounds $\frac{1}{4}$ and $4$, say; see equation~(17) in~\cite{BandeiraCMN:14}.
Notably, these bounds do not get worse as $d$ gets large.
It would be interesting if a version of Theorem~\ref{thm.main result} held for infinite groups, but we do not expect it to hold for infinite-dimensional inner product spaces.
Case in point, for $V=\ell^2$ with $G=\{\pm\operatorname{id}\}$, it was shown in~\cite{CahillCD:16} that for every choice of templates, the map is \textit{not} bilipschitz.

Another interesting phenomenon from finite-dimensional phase retrieval is that separating implies bilipschitz; see Lemma~16 and Theorem~18 in~\cite{BandeiraCMN:14} and Proposition~1.4 in~\cite{CahillCD:16}.
This suggests the following:

\begin{problem}
\label{prob.sep implies bilip}
Is every separating max filter bank $\Phi\colon \mathbb{R}^d/G\to\mathbb{R}^n$ bilipschitz?
\end{problem}

If the answer to Problem~\ref{prob.sep implies bilip} is ``yes,'' then Corollary~\ref{cor.2d templates suffice for finite groups} implies that $2d$ generic templates produce a bilipschitz max filter bank $\Phi\colon\mathbb{R}^d/G\to\mathbb{R}^{2d}$ whenever $G\leq\operatorname{O}(d)$ is finite.

Theorem~\ref{thm.main result} follows immediately from Lemmas~\ref{lem.bilipschitz no randomness} and~\ref{lem.random vectors have projective uniformity} below.
Our proof uses the following notion that was introduced in~\cite{AlexeevBFM:14}.
We say $\{z_i\}_{i=1}^n\in(\mathbb{R}^d)^n$ exhibits $(k,\delta)$-\textbf{projective uniformity} if
\[
s_k\{|\langle z_i,x\rangle|\}_{i=1}^n
\geq\delta\|x\|
\]
for every $x\in\mathbb{R}^d$, where $s_k\colon\mathbb{R}^n\to\mathbb{R}$ returns the $k$th smallest entry of the input.
In what follows, we denote $\|\{z_i\}_{i=1}^n\|_F:=(\sum_{i=1}^n\|z_i\|^2)^{1/2}$.

\begin{lemma}
\label{lem.bilipschitz no randomness}
Fix a finite subgroup $G\leq\operatorname{O}(d)$ and suppose $\{z_i\}_{i=1}^n\in(\mathbb{R}^d)^n$ exhibits $(\lceil\frac{n}{|G|^2}\rceil,\delta)$-projective uniformity.
Then the max filter bank $\Phi\colon\mathbb{R}^d/G\to\mathbb{R}^n$ with templates $\{z_i\}_{i=1}^n$ has lower Lipschitz bound $\delta$ and upper Lipschitz bound $\|\{z_i\}_{i=1}^n\|_F$.
\end{lemma}

\begin{proof}
The upper Lipschitz bound follows from Lemma~\ref{lem.max filter product properties}(g):
\[
\|\Phi([x])-\Phi([y])\|^2
=\sum_{i=1}^n|\llangle [z_i],[x]\rrangle-\llangle [z_i],[y]\rrangle|^2
\leq\|\{z_i\}_{i=1}^n\|_F^2\cdot d([x],[y])^2.
\]
For the lower Lipschitz bound, fix $x,y\in\mathbb{R}^d$ with $[x]\neq[y]$, and then for each $i\in\{1,\ldots,n\}$, select $g_i,h_i\in G$ such that $\llangle [z_i],[x]\rrangle=\langle z_i,g_ix\rangle$ and $\llangle [z_i],[y]\rrangle=\langle z_i,h_iy\rangle$.
Then
\begin{equation}
\label{eq.to lower bound}
\|\Phi([x])-\Phi([y])\|^2
=\sum_{i=1}^n\langle z_i,g_ix-h_iy\rangle^2
\geq\bigg(\sum_{i=1}^n\langle z_i,\tfrac{g_ix-h_iy}{\|g_ix-h_iy\|}\rangle^2\bigg)\cdot d([x],[y])^2,
\end{equation}
where the inequality follows from the bound $\|g_ix-h_iy\|\geq d([x],[y])$.
Next, consider the map $p\colon i\mapsto (g_i,h_i)$, and select $(g,h)\in G^2$ with the largest preimage.
By pigeonhole, we have $|p^{-1}(g,h)|\geq\lceil\frac{n}{|G|^2}\rceil=:k$, and so
\[
\sum_{i=1}^n\langle z_i,\tfrac{g_i x-h_i y}{\|g_i x-h_i y\|}\rangle^2
\geq \max_{i\in p^{-1}(g,h)}\langle z_i,\tfrac{g x-h y}{\|g x-h y\|}\rangle^2
\geq s_k\{\langle z_i,\tfrac{g x-h y}{\|g x-h y\|}\rangle^2\}_{i=1}^n
\geq \delta^2.
\]
Combining with \eqref{eq.to lower bound} gives the result.
\end{proof}

The following lemma gives that random templates exhibit projective uniformity.

\begin{lemma}[cf.\ Lemma~6.9 in~\cite{AlexeevBFM:14}]
\label{lem.random vectors have projective uniformity}
Select $p\in(0,1)$ and take
\begin{equation}
\label{eq.choice for delta}
\delta
:=(\tfrac{\pi}{128}\cdot\tfrac{p^2}{2d+3\log(4/p)})^{1/2},
\qquad
n\geq\tfrac{12d}{p}\log(\tfrac{2}{\delta}+1).
\end{equation}
Draw independent random vectors $z_1,\ldots,z_n\sim\mathsf{Unif}(S^{d-1})$.
Then $\{z_i\}_{i=1}^n$ exhibits $(\lceil pn\rceil,\delta)$-projective uniformity with probability $\geq1-e^{-pn/12}$.
\end{lemma}

\begin{proof}
Put $k:=\lceil pn\rceil$, let $\mathcal{E}$ denote the failure event that $\{z_i\}_{i=1}^n$ does not have $(k,\delta)$-projective uniformity, and let $N_\delta$ denote a $\delta$-net of $S^{d-1}$ of minimum size.
Note that if $v$ is within $\delta$ of $x$, then for every $z_i$, it holds that
\[
|\langle z_i,v\rangle|
\leq |\langle z_i,x\rangle|+|\langle z_i,v-x\rangle|
\leq |\langle z_i,x\rangle|+\|v-x\|
\leq |\langle z_i,x\rangle|+\delta.
\]
Thus, we may pass to the $\delta$-net to get
\begin{align*}
\mathbb{P}(\mathcal{E})
&=\mathbb{P}\Big\{\text{ $\exists x\in S^{d-1}$ s.t.\ $s_k\{|\langle z_i,x\rangle|\}_{i=1}^n<\delta$ }\Big\}\\
&\leq\mathbb{P}\Big\{\text{ $\exists v\in N_\delta$ s.t.\ $s_k\{|\langle z_i,v\rangle|\}_{i=1}^n<2\delta$ }\Big\}\\
&\leq |N_\delta|\cdot\mathbb{P}\Big\{s_k\{|\langle z_i,e_1\rangle|\}_{i=1}^n<2\delta\Big\}
=|N_\delta|\cdot\mathbb{P}\Big\{\sum_{i=1}^n\mathbf{1}_{\{|\langle z_i,e_1\rangle|<2\delta\}}\geq k\Big\},
\end{align*}
where the second inequality applies the union bound and the rotation invariance of the distribution $\mathsf{Unif}(S^{d-1})$.
A standard volume comparison argument gives $|N_\delta|\leq(\frac{2}{\delta}+1)^d$. 
The final probability concerns a sum of independent Bernoulli variables with some success probability $q=q(d,\delta)$, which can be estimated using the multiplicative Chernoff bound:
\[
\mathbb{P}\Big\{\sum_{i=1}^n\mathbf{1}_{\{|\langle z_i,e_1\rangle|<2\delta\}}\geq k\Big\}
\leq \mathbb{P}\Big\{\sum_{i=1}^n\mathbf{1}_{\{|\langle z_i,e_1\rangle|<2\delta\}}\geq pn\Big\}
\leq \exp(-\tfrac{(p-q)^2}{p+q}\cdot n),
\]
provided $p>q$.
Next, we verify that $q(d,\delta)\leq\frac{p}{2}$.
Denoting $g\sim\mathsf{N}(0,I_d)$, we have
\begin{align*}
q
:=\mathbb{P}\{|\langle \tfrac{g}{\|g\|},e_1\rangle|<2\delta\}
&\leq \inf_{t>0}\Big(\mathbb{P}\{|\langle g,e_1\rangle|<2\delta t\}+\mathbb{P}\{\|g\|^2>t^2\}\Big)\\
&\leq \inf_{t\geq\sqrt{2d}}\Big(\sqrt{\tfrac{2}{\pi}}\cdot2\delta t+e^{-(t^2-2d)/3}\Big),
\end{align*}
where the final inequality uses the facts that $|\langle g,e_1\rangle|$ has half-normal distribution and $\|g\|^2$ has chi-squared distribution with $d$ degrees of freedom.
We select $t:=(2d+3\log(\frac{4}{p}))^{1/2}$ so that the second term equals $\frac{p}{4}$, and then our choice \eqref{eq.choice for delta} for $\delta$ ensures that the first term equals $\frac{p}{4}$.
Overall, we have
\begin{align*}
\mathbb{P}(\mathcal{E})
&\leq |N_\delta|\cdot\mathbb{P}\Big\{\sum_{i=1}^n\mathbf{1}_{\{|\langle z_i,e_1\rangle|<2\delta\}}\geq k\Big\}\\
&\leq (\tfrac{2}{\delta}+1)^d\cdot \exp(-\tfrac{(p-q)^2}{p+q}\cdot n)
\leq \exp(d\log(\tfrac{2}{\delta}+1)-\tfrac{pn}{6})
\leq e^{-pn/12},
\end{align*}
where the last step applied our assumption that $n\geq\frac{12d}{p}\log(\frac{2}{\delta}+1)$.
\end{proof}

\subsection{Mallat-type stability to diffeomorphic distortion}

In this subsection, we focus on the case in which $V=L^2(\mathbb{R}^d)$ and $G$ is the group of translation operators $T_a$ defined by $T_af(x):=f(x-a)$ for $a\in \mathbb{R}^d$.
Given a template $h\in L^2(\mathbb{R}^d)$, the corresponding max filter is
\[
\llangle [h],[f]\rrangle
=\sup_{a\in\mathbb{R}^d}\langle h,T_af \rangle
=\sup_{a\in\mathbb{R}^d}\int_{\mathbb{R}^d}h(x)f(x-a)dx
=\sup_{a\in\mathbb{R}^d}(Rh\star f)(a),
\]
where $R$ denotes the reversal operator defined by $Rh(x):=h(-x)$ and $\star$ denotes convolution.
(Of course, the supremum of a member of $L^2(\mathbb{R}^d)$ is not well defined, but $Rh\star f$ is continuous since $Rh, f\in L^2(\mathbb{R}^d)$.)

Our motivation for this setting stems from image analysis, in which case $d=2$.
For a familiar example, consider the task of classifying handwritten digits.
Intuitively, each class is translation invariant, and so it makes sense to treat images as members of $V/G$.
In addition, images that are slight elastic distortions of each other should be sent to nearby points in the feature domain.
The fact that image classification is invariant to such distortions has been used to augment the MNIST training set and boost classification performance~\cite{SimardSP:03}.
Instead of using data augmentation to learn distortion-invariant features, it is desirable to restrict to feature maps that already exhibit distortion invariance.
(Indeed, such feature maps would require fewer parameters to train.)
This compelled Mallat to introduce his \textit{scattering transform}~\cite{Mallat:12}, which has since played an important role in the theory of invariant machine learning~\cite{BrunaM:11,BrunaM:13,Waldspurger:17,GaoWH:19,PerlmutterGWH:19}.
Mallat used the following formalism to analyze the stability of the scattering transform to distortion.

Given a diffeomorphism $g\in C^1(\mathbb{R}^d)$, we consider the corresponding distortion operator $L_g$ defined by $L_gf(x):=f(g^{-1}(x))$.
It will be convenient to interact with the vector field $\tau:=\operatorname{id}-g^{-1}\in C^1(\mathbb{R}^d)$, since $L_gf(x)=f(x-\tau(x))$.
For example, if $\tau(x)=a$ for every $x\in\mathbb{R}^d$, then $L_g$ is translation by $a$.
In what follows, $J\tau(x)\in\mathbb{R}^{d\times d}$ denotes the Jacobian matrix of $\tau$ at $x$.

\begin{theorem}
\label{thm.mallat bound}
Take any continuously differentiable $h\in L^2(\mathbb{R}^d)$ for which
\begin{equation}
\label{eq.bounded decay}
(1+\|x\|_2)^{(d+1)/2} \cdot h(x)
\qquad
\text{and}
\qquad
(1+\|x\|_2)^{(d+3)/2}\cdot\nabla h(x)
\end{equation}
are bounded.
There exists $C(h)>0$ such that for every $f\in L^2(\mathbb{R}^d)$ and every diffeomorphism $g\in C^1(\mathbb{R}^d)$ for which $\tau:=\operatorname{id}-g^{-1}$ satisfies $\sup_{x\in\mathbb{R}^d}\|J\tau(x)\|_{2\to2}\leq\frac{1}{2}$, it holds that
\[
|\llangle[h],[f]\rrangle-\llangle[h],[L_gf]\rrangle|
\leq C(h)\cdot\|f\|_{L^2(\mathbb{R}^d)}\cdot\sup_{x\in\mathbb{R}^d}\|J\tau(x)\|_{2\to2}.
\]
\end{theorem}

This matches Mallat's bound~\cite{Mallat:12} on the stability of the scattering transform to diffeomorphic distortion, except our bound has no Hessian term.
The proof of Theorem~\ref{thm.mallat bound} follows almost immediately from the following modification of Lemma~E.1 in~\cite{Mallat:12}, which bounds the commutator between the filter and the distortion by the magnitude of the distortion:

\begin{lemma}
\label{lem.commutator bound}
Take $h\in L^2(\mathbb{R}^d)$ as in Theorem~\ref{thm.mallat bound}, and consider the linear operator $Z_h$ defined by $Z_hf:=h\star f$.
There exists $C(h)>0$ such that for every diffeomorphism $g\in C^1(\mathbb{R}^d)$ for which $\tau:=\operatorname{id}-g^{-1}$ satisfies $\sup_{x\in\mathbb{R}^d}\|J\tau(x)\|_{2\to2}\leq\frac{1}{2}$, it holds that
\[
\|L_gZ_h-Z_hL_g\|_{L^2(\mathbb{R}^d)\to L^\infty(\mathbb{R}^d)}
\leq C(h)\cdot\sup_{x\in\mathbb{R}^d}\|J\tau(x)\|_{2\to2}.
\]
\end{lemma}

Assuming Lemma~\ref{lem.commutator bound} for the moment, we can prove Theorem~\ref{thm.mallat bound}.

\begin{proof}[Proof of Theorem~\ref{thm.mallat bound}]
The change of variables $a=g^{-1}(a')$ gives
\begin{align*}
|\llangle[h],[f]\rrangle-\llangle[h],[L_gf]\rrangle|
&=\Big|\sup_{a\in\mathbb{R}^d}(Z_{Rh}f)(a)-\sup_{b\in\mathbb{R}^d}(Z_{Rh}L_gf)(b)\Big|\\
&=\Big|\sup_{a'\in\mathbb{R}^d}(L_gZ_{Rh}f)(a')-\sup_{b\in\mathbb{R}^d}(Z_{Rh}L_gf)(b)\Big|\\
&\leq\|L_gZ_{Rh}f-Z_{Rh}L_gf\|_{L^\infty(\mathbb{R}^d)}\\
&\leq\|L_gZ_{Rh}-Z_{Rh}L_g\|_{L^2(\mathbb{R}^d)\to L^\infty(\mathbb{R}^d)}\cdot\|f\|_{L^2(\mathbb{R}^d)},
\end{align*}
and so the result follows from Lemma~\ref{lem.commutator bound}.
\end{proof}

The rest of this section proves Lemma~\ref{lem.commutator bound}.
Our proof follows some of the main ideas in the proof of Lemma~E.1 in~\cite{Mallat:12}.

\begin{proof}[Proof of Lemma~\ref{lem.commutator bound}]
Denote $K:=Z_h-L_gZ_hL_g^{-1}$.
Then $L_gZ_h-Z_hL_g=-KL_g$, and so
\[
\|L_gZ_h-Z_hL_g\|_{L^2(\mathbb{R}^d)\to L^\infty(\mathbb{R}^d)}
=\|KL_g\|_{L^2(\mathbb{R}^d)\to L^\infty(\mathbb{R}^d)}
\leq\|K\|_{L^2(\mathbb{R}^d)\to L^\infty(\mathbb{R}^d)}\|L_g\|_{L^2(\mathbb{R}^d)\to L^2(\mathbb{R}^d)}.
\]
We first bound the second factor.
For $f\in L^2(\mathbb{R}^d)$, a change of variables gives
\[
\|L_gf\|_{L^2(\mathbb{R}^d)}^2
=\int_{\mathbb{R}^d}f(g^{-1}(x))^2dx
=\int_{\mathbb{R}^d}f(u)^2|\operatorname{det}(Jg(u))|du
\leq\sup_{u\in\mathbb{R}^d}|\operatorname{det}(Jg(u))|\cdot\|f\|_{L^2(\mathbb{R}^d)}^2.
\]
For $x\in\mathbb{R}^d$, the fact that $\|J\tau(x)\|_{2\to2}\leq\frac{1}{2}$ implies
\[
\operatorname{det}(Jg^{-1}(x))
=\operatorname{det}(I_d-J\tau(x))
\geq(1-\|J\tau(x)\|_{2\to2})^d
\geq 2^{-d},
\]
and so combining with the above estimate gives
\begin{equation}
\label{eq.bound on L_g}
\|L_g\|_{L^2(\mathbb{R}^d)\to L^2(\mathbb{R}^d)}
\leq\sup_{u\in\mathbb{R}^d}|\operatorname{det}(Jg(u))|
=\sup_{x\in\mathbb{R}^d}|\operatorname{det}(Jg^{-1}(x))|^{-1}
\leq 2^d.
\end{equation}
It remains to bound $\|K\|_{L^2(\mathbb{R}^d)\to L^\infty(\mathbb{R}^d)}$.
To this end, one may verify that $K$ can be expressed as $Kf(x)=\int_{\mathbb{R}^d}k(x,u)f(u)du$, where the kernel $k$ is defined by
\[
k(x,u)
:=h(x-u)-\operatorname{det}(I_d-J\tau(u))\cdot h(x-u-\tau(x)+\tau(u)).
\]
We will bound the $L^2$ norms of every $k(x,\cdot)$ and $k(\cdot,u)$, and then appeal to Young's inequality for integral operators to bound $\|K\|_{L^2(\mathbb{R}^d)\to L^\infty(\mathbb{R}^d)}$.
We decompose $k=k_1+k_2+k_3$, where
\begin{align*}
k_1(x,u)
&:=h(x-u)-h\big((I_d-J\tau(u))(x-u)\big),\\
k_2(x,u)
&:=\big(1-\operatorname{det}(I_d-J\tau(u))\big)\cdot h\big((I_d-J\tau(u))(x-u)\big),\\
k_3(x,u)
&:=\operatorname{det}(I_d-J\tau(u))\cdot\Big(h\big((I_d-J\tau(u))(x-u)\big)-h(x-u-\tau(x)+\tau(u))\Big).
\end{align*}
First, we analyze $k_1$.
Letting $p_1\colon[0,1]\to\mathbb{R}^d$ denote the parameterized line segment of constant velocity from $(I_d-J\tau(u))(x-u)$ to $x-u$, we have
\[
k_1(x,u)
=\int_0^1\nabla h(p_1(t))\cdot J\tau(u)(x-u)~dt,
\]
and so
\begin{align}
\nonumber
|k_1(x,u)|
&\leq \sup_{t\in[0,1]}\|\nabla h(p_1(t))\|_2\cdot\|J\tau(u)(x-u)\|_2\\
\label{eq.k1 bound}
&\leq \sup_{t\in[0,1]}\|\nabla h(p_1(t))\|_2\cdot\sup_{z\in\mathbb{R}^d}\|J\tau(z)\|_{2\to2}\cdot\|x-u\|_2.
\end{align}
To bound the first factor, let $C_\infty(h)>0$ denote a simultaneous bound on the absolute value and $2$-norm of \eqref{eq.bounded decay}.
To use this, we bound $\inf_{t\in[0,1]}\|p_1(t)\|_2$ from below:
\begin{align*}
\|p_1(t)\|_2
=\|x-u+tJ\tau(u)(x-u)\|_2
&\geq\|x-u\|_2-t\|J\tau(u)(x-u)\|_2\\
&\geq(1-t\sup_{z\in\mathbb{R}^d}\|J\tau(u)\|_{2\to2})\cdot\|x-u\|_2
\geq\tfrac{1}{2}\|x-u\|_2.
\end{align*}
Then
\[
\sup_{t\in[0,1]}\|\nabla h(p_1(t))\|_2
\leq\sup_{t\in[0,1]}\frac{C_\infty(h)}{(1+\|p_1(t)\|_2)^{(d+3)/2}}
\leq\frac{C_\infty(h)}{(1+\frac{1}{2}\|x-u\|_2)^{(d+3)/2}},
\]
which allows us to further bound \eqref{eq.k1 bound}:
\begin{equation}
\label{eq.final bound on k1}
|k_1(x,u)|
\leq \frac{C_\infty(h)}{(1+\frac{1}{2}\|x-u\|_2)^{(d+3)/2}}\cdot\sup_{z\in\mathbb{R}^d}\|J\tau(z)\|_{2\to2}\cdot\|x-u\|_2.
\end{equation}
Next, we analyze $k_2$.
Since $\|J\tau(x)\|_{2\to2}\leq\frac{1}{2}$ by assumption, Bernoulli's inequality gives
\[
1-\operatorname{det}(I_d-J\tau(x))
\leq 1-(1-\|J\tau(x)\|_{2\to2})^d
\leq d\|J\tau(x)\|_{2\to2}
\leq d\sup_{z\in\mathbb{R}^d}\|J\tau(z)\|_{2\to2}.
\]
Also, the convexity bound $(1+t)^d\leq 1+(2^d-1)t$ for $t\in[0,1]$ implies
\begin{align*}
1-\operatorname{det}(I_d-J\tau(x))
&\geq1-\|I_d-J\tau(x)\|_{2\to2}^d\\
&\geq1-(1+\|J\tau(x)\|_{2\to2})^d
\geq-2^d\sup_{z\in\mathbb{R}^d}\|J\tau(z)\|_{2\to2}.
\end{align*}
Furthermore, we have
\[
\|(I_d-J\tau(x))(x-u)\|_2
\geq\|x-u\|_2-\|J\tau(u)\|_{2\to2}\|x-u\|_2
\geq\tfrac{1}{2}\|x-u\|_2,
\]
and so
\begin{align}
\nonumber
|k_2(x,u)|
&\leq \big|1-\operatorname{det}(I_d-J\tau(u))\big|\cdot \big|h\big((I_d-J\tau(u))(x-u)\big)\big|\\
\label{eq.final bound on k2}
&\leq 2^d\sup_{z\in\mathbb{R}^d}\|J\tau(z)\|_{2\to2}\cdot \frac{C_\infty(h)}{(1+\frac{1}{2}\|x-u\|_2)^{(d+1)/2}}.
\end{align}
Finally, we analyze $k_3$.
Put
\[
r
:=x-u-\tau(x)+\tau(u),
\qquad
s
:=\tau(x)-\tau(u)-J\tau(u)(x-u).
\]
Then letting $p_2\colon[0,1]\to\mathbb{R}^d$ denote the parameterized line segment of constant velocity from $r$ to $r+s$, we have
\[
k_3(x,u)
=\operatorname{det}(I_d-J\tau(u))\cdot\big(h(r+s)-h(r)\big)
=\operatorname{det}(I_d-J\tau(u))\int_0^1\nabla h(p_2(t))\cdot s~dt,
\]
and so
\begin{equation}
\label{eq.bound on k3}
|k_3(x,u)|
\leq|\operatorname{det}(I_d-J\tau(u))|\cdot\sup_{t\in[0,1]}\|\nabla h(p_2(t))\|_2\cdot \|s\|_2.
\end{equation}
For the first factor of \eqref{eq.bound on k3}, we have $|\operatorname{det}(I_d-J\tau(u))|\leq\|I_d-J\tau(u)\|_{2\to2}^d\leq (3/2)^d$.
To bound the second factor of \eqref{eq.bound on k3}, we use our bound $C_\infty(h)>0$ on \eqref{eq.bounded decay}.
To do so, we bound $\inf_{t\in[0,1]}\|p_2(t)\|_2$ from below.
First, we note that
\[
\tau(x)-\tau(u)
=\int_0^1J\tau(p_3(t))(x-u)dt,
\]
where $p_3\colon[0,1]\to\mathbb{R}^d$ is the parameterized line segment of constant velocity from $u$ to $x$.
Thus,
\begin{equation}
\label{eq.bound diff of taus}
\|\tau(x)-\tau(u)\|_2
\leq\int_0^1 \|J\tau(p_3(t))\|_{2\to2}\|x-u\|_2dt
\leq\sup_{z\in\mathbb{R}^d}\|J\tau(z)\|_{2\to2}\cdot\|x-u\|_2,
\end{equation}
and so for each $t\in[0,1]$, we have
\begin{align*}
\|p_2(t)\|_2
=\|r+ts\|_2
&=\|(x-u)-(1-t)(\tau(x)-\tau(u))-tJ\tau(u)(x-u)\|_2\\
&\geq\|x-u\|_2-(1-t)\|\tau(x)-\tau(u)\|-t\|J\tau(u)\|_{2\to2}\|x-u\|_2\\
&\geq(1-\sup_{z\in\mathbb{R}^d}\|J\tau(z)\|_{2\to2})\cdot\|x-u\|_2
\geq\tfrac{1}{2}\|x-u\|_2.
\end{align*}
Overall, we have
\[
\sup_{t\in[0,1]}\|\nabla h(p_2(t))\|_2
\leq\sup_{t\in[0,1]}\frac{C_\infty(h)}{(1+\|p_2(t)\|_2)^{(d+3)/2}}
\leq\frac{C_\infty(h)}{(1+\frac{1}{2}\|x-u\|_2)^{(d+3)/2}}.
\]
Finally, we apply \eqref{eq.bound diff of taus} to bound the third factor of \eqref{eq.bound on k3}:
\[
\|s\|_2
\leq\|\tau(x)-\tau(u)\|_2+\|J\tau(u)(x-u)\|_2
\leq2\cdot\sup_{z\in\mathbb{R}^d}\|J\tau(z)\|_{2\to2}\cdot\|x-u\|_2.
\]
We combine these estimates to obtain the following bound on \eqref{eq.bound on k3}:
\begin{equation}
\label{eq.final bound on k3}
|k_3(x,u)|
\leq (3/2)^d \cdot \frac{C_\infty(h)}{(1+\frac{1}{2}\|x-u\|_2)^{(d+3)/2}} \cdot 2\cdot\sup_{z\in\mathbb{R}^d}\|J\tau(z)\|_{2\to2}\cdot\|x-u\|_2.
\end{equation}
Finally, \eqref{eq.final bound on k1}, \eqref{eq.final bound on k2}, and \eqref{eq.final bound on k3} together imply
\begin{align*}
|k(x,u)|
&\leq|k_1(x,u)|+|k_2(x,u)|+|k_3(x,u)|\\
&\leq C_\infty(h)\cdot\sup_{z\in\mathbb{R}^d}\|J\tau(z)\|_{2\to2}\cdot\bigg(\frac{2^d}{(1+\frac{1}{2}\|x-u\|_2)^{(d+1)/2}}
+\frac{(2(\frac{3}{2})^d+1)\|x-u\|_2}{(1+\frac{1}{2}\|x-u\|_2)^{(d+3)/2}}\bigg).
\end{align*}
Importantly, this is a bounded function of $x-u$ that decays like $\|x-u\|_2^{-(d+1)/2}$.
By integrating the square, this simultaneously bounds the $L^2$ norm of every $k(x,\cdot)$ and $k(\cdot,u)$ by a quantity of the form $C_0(h)\cdot\sup_{z\in\mathbb{R}^d}\|J\tau(z)\|_{2\to2}$.
By Young's inequality for integral operators (see Theorem~0.3.1 in~\cite{Sogge:17}, for example), it follows that
\[
\|K\|_{L^2(\mathbb{R}^d)\to L^\infty(\mathbb{R}^d)}
\leq C_0(h)\cdot\sup_{z\in\mathbb{R}^d}\|J\tau(z)\|_{2\to2}.
\]
Combining with \eqref{eq.bound on L_g} then gives the result with $C(h):=2^d\cdot C_0(h)$.
\end{proof}

\section{Template selection for classification}
\label{sec.classification}

\subsection{Classifying characteristic functions}

In this subsection, we focus on the case in which $V=L^2(\mathbb{R}^d)$ and $G$ is the group of translation operators.
Suppose we have $k$ distinct $G$-orbits of indicator functions of compact subsets of $\mathbb{R}^d$.
According to the following result, there is a simple classifier based on a size-$k$ max filter bank that correctly classifies these orbits.
(This cartoon setting enjoys precursors in~\cite{CaulfieldM:69,CaulfieldH:80}.)

\begin{theorem}
\label{thm.classifying indicator functions}
Given compact sets $S_1,\ldots,S_k\subseteq\mathbb{R}^d$ of positive measure satisfying
\[
[\mathbf{1}_{S_i}]\neq [\mathbf{1}_{S_j}]
\qquad
\text{whenever}
\qquad
i\neq j,
\]
there exist templates $z_1,\ldots, z_k\in L^2(\mathbb{R}^d)$ satisfying
\[
\llangle [z_i],[\mathbf{1}_{S_j}]\rrangle<\llangle [z_i],[\mathbf{1}_{S_i}]\rrangle
\qquad
\text{whenever}
\qquad
i\neq j.
\]
\end{theorem}

\begin{proof}
By compactness, there exists $r>0$ such that every $S_i$ is contained in the closed ball centered at the origin with radius $r$.
For reasons that will become apparent later, we take $B$ to be the closed ball centered at the origin with radius $3r$, and we define $z_i:=\mathbf{1}_{S_i}-\mathbf{1}_{B\setminus S_i}$.
Then for every $i$ and $j$, it holds that
\begin{equation}
\label{eq.bound on indicator template}
\llangle [z_i],[\mathbf{1}_{S_j}]\rrangle
=\sup_{a\in\mathbb{R}^d}\langle \mathbf{1}_{S_i}-\mathbf{1}_{B\setminus S_i},T_a\mathbf{1}_{S_j}\rangle
\leq \sup_{a\in\mathbb{R}^d}\langle \mathbf{1}_{S_i},T_a\mathbf{1}_{S_j}\rangle
=\llangle [\mathbf{1}_{S_i}],[\mathbf{1}_{S_j}]\rrangle.
\end{equation}
In the special case where $j=i$, this implies
\[
|S_i|
=\langle z_i,\mathbf{1}_{S_i}\rangle
\leq\llangle [z_i],[\mathbf{1}_{S_i}]\rrangle
\leq\llangle [\mathbf{1}_{S_i}],[\mathbf{1}_{S_i}]\rrangle
=|S_i|,
\]
and so $\llangle [z_i],[\mathbf{1}_{S_i}]\rrangle=|S_i|$.
We consider all $j\neq i$ in two cases.

\medskip
\noindent
\textbf{Case I:} $j\neq i$ and $|S_j|\leq|S_i|$.
Considering \eqref{eq.bound on indicator template}, it suffices to bound $\llangle [\mathbf{1}_{S_i}],[\mathbf{1}_{S_j}]\rrangle$.
Letting $R$ denote the reversal operator defined by $Rf(x):=f(-x)$, then
\[
\llangle [\mathbf{1}_{S_i}],[\mathbf{1}_{S_j}]\rrangle
=\sup_{a\in\mathbb{R}^d}(R\mathbf{1}_{S_i}\star\mathbf{1}_{S_j})(a).
\]
Since $R\mathbf{1}_{S_i},\mathbf{1}_{S_j}\in L^2(\mathbb{R}^d)$, it holds that the convolution $R\mathbf{1}_{S_i}\star\mathbf{1}_{S_j}$ is continuous, and since $S_i$ and $S_j$ are compact, the convolution has compact support.
Thus, the extreme value theorem gives that the convolution achieves its supremum, meaning there exists $a\in\mathbb{R}^d$ such that
\begin{equation}
\label{eq.indicators 2}
\llangle [\mathbf{1}_{S_i}],[\mathbf{1}_{S_j}]\rrangle
=\langle \mathbf{1}_{S_i},T_a\mathbf{1}_{S_j}\rangle
=|S_i\cap (S_j+a)|.
\end{equation}
Next, the assumptions $|S_j|\leq|S_i|$ and $[\mathbf{1}_{S_i}]\neq[\mathbf{1}_{S_j}]$ together imply
\begin{equation}
\label{eq.indicators 3}
|S_i\cap (S_j+a)|
<|S_i|.
\end{equation}
Indeed, equality in the bound $|S_i\cap (S_j+a)|\leq |S_i|$ is only possible if $S_i\subseteq S_j+a$ (modulo null sets), but since $|S_j|\leq|S_i|$ by assumption, this requires $S_i=S_j+a$ (modulo null sets), which violates the assumption $[\mathbf{1}_{S_i}]\neq[\mathbf{1}_{S_j}]$.
Overall, we combine \eqref{eq.bound on indicator template}, \eqref{eq.indicators 2}, and \eqref{eq.indicators 3} to get
\[
\llangle [z_i],[\mathbf{1}_{S_j}]\rrangle
\leq\llangle [\mathbf{1}_{S_i}],[\mathbf{1}_{S_j}]\rrangle
=|S_i\cap (S_j+a)|
<|S_i|
=\llangle [z_i],[\mathbf{1}_{S_i}]\rrangle.
\]

\medskip
\noindent
\textbf{Case II:} $j\neq i$ and $|S_j|\geq|S_i|$.
If $\llangle [z_i],[\mathbf{1}_{S_j}]\rrangle\leq0$, then 
\[
\llangle [z_i],[\mathbf{1}_{S_j}]\rrangle
\leq0
<|S_i|
=\llangle [z_i],[\mathbf{1}_{S_i}]\rrangle,
\]
and so we are done.
Now suppose $\llangle [z_i],[\mathbf{1}_{S_j}]\rrangle>0$.
Considering
\[
\llangle [z_i],[\mathbf{1}_{S_j}]\rrangle
=\sup_{a\in\mathbb{R}^d}(Rz_i\star\mathbf{1}_{S_j})(a),
\]
then by continuity and compactness, the extreme value theorem produces $a\in\mathbb{R}^d$ such that
\begin{equation}
\label{eq.indicators 4}
\llangle [z_i],[\mathbf{1}_{S_j}]\rrangle
=\langle z_i,T_a\mathbf{1}_{S_j}\rangle
=|S_i\cap(S_j+a)|-|(B\setminus S_i)\cap(S_j+a)|.
\end{equation}
Since $\llangle [z_i],[\mathbf{1}_{S_j}]\rrangle>0$, it follows that $|S_i\cap(S_j+a)|>0$, i.e., $S_i\cap(S_j+a)$ is nonempty, which in turn implies $S_j+a\subseteq B$.
(This is why we defined $B$ to have radius $3r$.)
As before, $|S_j|\geq|S_i|$ and $[\mathbf{1}_{S_i}]\neq [\mathbf{1}_{S_j}]$ together give $|S_i\cap(S_j+a)|<|S_j|$.
Thus,
\begin{equation}
\label{eq.indicators 5}
|(B\setminus S_i)\cap(S_j+a)|
=|B\cap(S_j+a)|-|S_i\cap(S_j+a)|
=|S_j+a|-|S_i\cap(S_j+a)|
>0.
\end{equation}
We combine \eqref{eq.indicators 4} and \eqref{eq.indicators 5} to get
\[
\llangle [z_i],[\mathbf{1}_{S_j}]\rrangle
=|S_i\cap(S_j+a)|-|(B\setminus S_i)\cap(S_j+a)|
<|S_i\cap(S_j+a)|
\leq|S_i|
=\llangle [z_i],[\mathbf{1}_{S_i}]\rrangle,
\]
as claimed.
\end{proof}

For each $i$, assume $S_i$ is translated so that it is contained in the smallest possible ball centered at the origin, and let $r_i$ denote the radius of this ball.
The proof of Theorem~\ref{thm.classifying indicator functions} gives that each template $z_i$ is supported in a closed ball of radius $R:=3\max_i r_i$.
The fact that these templates are localized bears some consequence for certain \textit{image articulation manifolds}~\cite{DonohoG:05}.
In particular, for each $\pi\colon\{1,\ldots,k\}\to\mathbb{N}\cup\{0\}$, let $M_\pi \subseteq L^2(\mathbb{R}^d)$ denote the manifold of images of the form
\[
\sum_{i=1}^k \sum_{j=1}^{\pi(i)} T_{a(i,j)}\mathbf{1}_{S_i},
\qquad
\text{where}
\qquad
\|a(i,j)-a(i',j')\|>4R
\quad
\forall (i,j)\neq(i',j').
\]
Thanks to the $4R$ spacing, each translate of each template interacts with at most one component of the image, and so for every $f\in M_\pi$, it holds that
\[
\llangle [z_i],[f]\rrangle
=\max_{i':\pi(i')>0}\llangle [z_i],[\mathbf{1}_{S_{i'}}]\rrangle.
\]
In particular, the same max filter bank can be used to determine the support of $\pi$.
As an example, if some multiset of characters are typed on a page in a common font and with sufficient separation, then the max filter bank from Theorem~\ref{thm.classifying indicator functions} that distinguishes the characters can be used to determine which ones appear on the page.

\subsection{Classifying mixtures of stationary processes}

In this subsection, we focus on the case in which $V=\mathbb{R}^n$ and $G\cong C_n$ is the group of circular translation operators.
A natural $G$-invariant probability distribution is a multivariate Gaussian with mean zero and circulant covariance, and so we consider the task of classifying a mixture of such distributions.
One-dimensional textures can be modeled in this way, especially if the covariance matrix has a small bandwidth so that distant pixels are statistically independent.
A standard approach for this problem is to estimate the first- and second-order moments given a random draw.
As we will soon see, one can alternatively classify with high accuracy by thresholding a single max filter.
In what follows, we make use of \textit{Thompson's part metric} on the set of positive definite matrices:
\[
d_\infty(A,B)
:=\|\log(A^{-1/2}BA^{-1/2})\|_{2\to2}.
\]
We also let $A_k$ denote the leading $k\times k$ principal submatrix of $A$.

\begin{theorem}
\label{thm.binary gmm classification}
Fix $C>\log2$, take $n,w\in\mathbb{N}$ such that $k:=\lfloor \sqrt{n/2}\rfloor\geq w$, and consider any positive definite $A,B\in\mathbb{R}^{n\times n}$ that are circulant with bandwidth $w$ and satisfy 
\begin{equation}
\label{eq.thompson distance bound}
d_\infty(A_k,B_k)\geq C.
\end{equation}
There exists $z\in\mathbb{R}^n$ supported on an interval of length $k$ and a threshold $\theta\in\mathbb{R}$ such that for every mixture $\mathsf{M}$ of $\mathsf{N}(0,A)$ and $\mathsf{N}(0,B)$, then given $x\sim\mathsf{M}$, the comparison
\[
\llangle[z],[x]\rrangle
\gtrless \theta
\]
correctly classifies the latent mixture component of $x$ with probability $1-o_{n\to\infty;C}(1)$.
\end{theorem}

\begin{proof}
First, we observe that $d_\infty(A_k,B_k)\geq C$ is equivalent to
\[
\max\Big\{\lambda_{\mathrm{max}}(A_k^{-1/2}B_k A_k^{-1/2}),\lambda_{\mathrm{max}}(B_k^{-1/2}A_k B_k^{-1/2})\Big\}
\geq e^C.
\]
Without loss of generality, we may assume $\lambda_{\mathrm{max}}(A_k^{-1/2}B_k A_k^{-1/2})\geq e^C$.
Let $v\in\mathbb{R}^k$ denote a corresponding unit eigenvector of $A_k^{-1/2}B_k A_k^{-1/2}$, and define $z\in\mathbb{R}^n$ to be supported in its first $k$ entries as the subvector $z_k:=A_k^{-1/2}v$.
Then
\begin{equation}
\label{eq.variance quotient}
\frac{z^\top Bz}{z^\top Az}
=\frac{z_k^\top B_k z_k}{z_k^\top A_k z_k}
=\frac{v^\top A_k^{-1/2}B_k A_k^{-1/2}v}{v^\top v}
=\lambda_{\mathrm{max}}(A_k^{-1/2}B_k A_k^{-1/2})
\geq e^C.
\end{equation}
Consider $x_1\sim\mathsf{N}(0,A)$.
Then $\llangle [z],[x]\rrangle=\max_{a\in C_n}\langle T_az,x_1\rangle$, where each $\langle T_az,x_1\rangle$ has Gaussian distribution with mean zero and variance $(T_az)^\top A(T_az)$, which in turn equals $z^\top Az$ since $A$ is circulant.
Denoting $Z\sim\mathsf{N}(0,1)$, a union bound then gives
\[
\mathbb{P}\big\{\llangle [z],[x_1]\rrangle\geq t\big\}
\leq n\cdot\mathbb{P}\big\{(z^\top Az)^{1/2}\cdot Z\geq t\big\}
\leq ne^{-t^2/(2z^\top Az)}
\]
for $t\geq0$.
This failure probability is $o_{n\to\infty;c_1}(1)$ by taking $t:=\sqrt{c_1\cdot z^\top Az\cdot\log n}$ for any $c_1>2$.
Next, consider $x_2\sim\mathsf{N}(0,B)$, and take any subset $S\subseteq C_n$ consisting of $k$ members of $C_n$ of pairwise distance at least $2k$.
Then
\[
\llangle [z],[x_2]\rrangle
=\max_{a\in C_n}\langle T_az,x_2\rangle
\geq\max_{a\in S}\langle T_az,x_2\rangle.
\]
In this case, each $\langle T_az,x_2\rangle$ has Gaussian distribution with mean zero and variance $z^\top Bz$.
Furthermore, since $k\geq w$, these random variables have pairwise covariance zero, and since they are Gaussian, they are therefore independent.
For independent $Z_1,\ldots,Z_k\sim\mathsf{N}(0,1)$ and $t\geq0$, we have
\[
\mathbb{P}\Big\{\max_{i\in\{1,\ldots,k\}}Z_i\leq t\Big\}
=\mathbb{P}\{Z\leq t\}^k
\leq(1-\tfrac{t}{1+t^2}\cdot\tfrac{1}{\sqrt{2\pi}}e^{-t^2/2})^k,
\]
where the last step follows from a standard lower bound on the tail of a standard Gaussian distribution.
Take $t:=\sqrt{c_2\log k}$ to get 
\[
\mathbb{P}\Big\{\max_{i\in\{1,\ldots,k\}}Z_i\leq \sqrt{c_2\log k}\Big\}
\leq\Big(1-(\tfrac{\sqrt{c_2\log k}}{1+c_2\log k}\cdot\tfrac{1}{\sqrt{2\pi}}\cdot k^{1-c_2/2})\cdot\tfrac{1}{k}\Big)^k
\leq\exp\Big(-\tfrac{\sqrt{c_2\log k}}{1+c_2\log k}\cdot\tfrac{1}{\sqrt{2\pi}}\cdot k^{1-c_2/2}\Big),
\]
which is $o_{k\to\infty;c_2}(1)$ when $c_2<2$.
Overall, we simultaneously have
\[
\llangle [z],[x_1]\rrangle
<\theta_1
:=\sqrt{c_1\cdot z^\top Az\cdot\log n},
\qquad
\llangle [z],[x_2]\rrangle
>\theta_2
:=\sqrt{\tfrac{1}{2}c_2\cdot z^\top Bz\cdot\log n}
\]
with probability $1-o_{n\to\infty;c_1,c_2}(1)$, provided $c_1>2>c_2$.
We select $c_1:=2(\tfrac{e^C}{2})^{1/2}>2$ and $c_2:=2(\tfrac{e^C}{2})^{-1/2}<2$, and then \eqref{eq.variance quotient} implies
\[
\frac{\theta_2^2}{\theta_1^2}
=\frac{\frac{1}{2}c_2\cdot z^\top Bz}{c_1\cdot z^\top Az}
=\frac{1}{e^C}\cdot\frac{z^\top Bz}{z^\top Az}
\geq1.
\]
Thus, $\theta_1\leq\theta_2$, and so the result follows by taking $\theta:=(\theta_1+\theta_2)/2$.
\end{proof}

Given a mixture of $k$ Gaussians with covariance matrices that pairwise satisfy~\eqref{eq.thompson distance bound}, then we can perform multiclass classification by a one-vs-one reduction.
Indeed, the binary classifier in Theorem~\ref{thm.binary gmm classification} can be applied to all $\binom{k}{2}$ pairs of Gaussians, in which case we correctly classify the latent mixture component with high probability (provided $k$ is fixed and $n\to\infty$).

Interestingly, max filters can also distinguish between stationary processes with \textit{identical} first- and second-order moments.
For example, $x\sim\mathsf{Unif}(\{\pm1\}^n)$ and $y\sim\mathsf{N}(0,I_n)$ are both stationary with mean zero and identity covariance.
If $z\in\mathbb{R}^n$ is a standard basis element, then with high probability, it holds that
\[
\llangle[z],[x]\rrangle
\leq1
\leq \sqrt{\log n}
\leq\llangle[z],[y]\rrangle.
\]
This indicates that max filters incorporate higher-order moments.

\subsection{Subgradients}

In this subsection, we focus on the case in which $V=\mathbb{R}^d$ and $G$ is a closed subgroup of $\operatorname{O}(d)$.
The previous subsections carefully designed templates to classify certain data models.
For real-world classification problems, one is expected to train a classifier on a given training set of labeled data.
To do so, one selects a parameterized family of classifiers and then locally minimizes some notion of training loss over this family.
This is feasible provided the classifier is a differentiable function of the parameters.
We envision a classifier in which the first layer is a max filter bank, and so this subsection establishes how to differentiate a max filter with respect to the template.
This is made possible by convexity; see Lemma~\ref{lem.max filter product properties}(d).

Every convex function $f\colon \mathbb{R}^d\to\mathbb{R}$ has a \textbf{subdifferential} $\partial f\colon \mathbb{R}^d\to 2^{\mathbb{R}^d}$ defined by
\[
\partial f(x)
:=\{u\in \mathbb{R}^d : f(x+h) \geq f(x) + \langle h, u\rangle 
~~ \forall h\in \mathbb{R}^d \}.
\]
For a fixed $x\in\mathbb{R}^d$ and $u\in\partial f(x)$, it is helpful to interpret the graph of $z\mapsto f(x) + \langle z-x, u\rangle$ as a supporting hyperplane of the epigraph of $f$ at $x$.
For example, the absolute value function over $\mathbb{R}$ has the following subdifferential:
\[
\partial|\cdot|(x)
=\left\{\begin{array}{cl}
\{-1\} & \text{if } x<0, \\
{[}-1,1] & \text{if } x=0, \\
\{1\} & \text{if } x>0.
\end{array}\right.
\]
Indeed, this gives the slopes of the hyperplanes that support the epigraph of $|\cdot|$ at $x\in\mathbb{R}$.

Following~\cite{Watson:92}, we will determine the subdifferential of the max filtering map by first finding its \textit{directional derivatives}.
In general, the directional derivative of $f$ at $x$ in the direction of $v\neq0$ is given by
\[
f'(x;v)
:=\lim_{t\to0^+}\frac{f(x+tv)-f(x)}{t}
=\sup_{u\in\partial f(x)}\langle u,v\rangle.
\]
The second equality is a standard result; see for example Theorem~23.4 in~\cite{Rockafellar:70}.
It will be convenient to denote the set
\[
G(x,y)
:=\arg\max_{g\in G}\langle x,gy\rangle.
\]
Observe that $G(x,y)$ is a closed subset of $G$ since the map $g\mapsto\langle x,gy\rangle$ is continuous.

\begin{lemma}
\label{lem.directional derivative}
Suppose $G$ is a closed subgroup of $\operatorname{O}(d)$ and select $x,y,v\in\mathbb{R}^d$ with $v\neq0$.
Then
\[
\llangle[\cdot],[y]\rrangle'(x;v)
=\max_{g\in G(x,y)}\langle v,gy\rangle.
\]
\end{lemma}

\begin{proof}
For each $t>0$ and $g\in G(x,y)$, we have
\begin{equation}
\label{eq.directional derivative 1}
\llangle[x+tv],[y]\rrangle
\geq\langle x+tv,gy\rangle
=\langle x,gy\rangle+t\langle v,gy\rangle
=\llangle [x],[y]\rrangle+t\langle v,gy\rangle.
\end{equation}
For each $t>0$, select $g_t\in G(x+tv,y)$.
Then
\begin{equation}
\label{eq.directional derivative 2}
\llangle[x+tv],[y]\rrangle
=\langle x+tv,g_ty\rangle
=\langle x,g_ty\rangle+t\langle v,g_ty\rangle
\leq\llangle [x],[y]\rrangle+t\langle v,g_ty\rangle.
\end{equation}
We rearrange and combine \eqref{eq.directional derivative 1} and \eqref{eq.directional derivative 2} to get
\begin{equation}
\label{eq.directional derivative 3}
\max_{g\in G(x,y)}\langle v,gy\rangle
\leq \frac{\llangle[x+tv],[y]\rrangle-\llangle[x],[y]\rrangle}{t}
\leq \langle v,g_ty\rangle
\end{equation}
for all $t>0$.
Select a sequence $t_n\to0^+$ such that $g_{t_n}$ converges to some $g^\star\in G$ (by passing to a subsequence if necessary).
Then continuity implies
\[
\langle x+t_nv,g_{t_n}y\rangle
=\llangle [x+t_nv],[y]\rrangle
\to \llangle [x],[y]\rrangle.
\]
Furthermore,
\begin{align*}
|\langle x+t_nv,g_{t_n}y\rangle-\langle x,g^\star y\rangle|
&\leq|\langle x,g_{t_n}y\rangle-\langle x,g^\star y\rangle|+|\langle t_nv,g_{t_n}y\rangle|\\
&\leq \|x\|\|g_{t_n}-g^\star\|_{2\to2} \|y\|+t_n\|v\|\|y\|
\to0.
\end{align*}
It follows that $\langle x,g^\star y\rangle=\llangle [x],[y]\rrangle$, i.e., $g^\star\in G(x,y)$.
Taking limits of \eqref{eq.directional derivative 3} then gives
\[
\max_{g\in G(x,y)}\langle v,gy\rangle
\leq \lim_{t\to0^+}\frac{\llangle[x+tv],[y]\rrangle-\llangle[x],[y]\rrangle}{t}
\leq \lim_{t\to0^+}\langle v,g_ty\rangle
=\langle v,g^\star y\rangle
\leq \max_{g\in G(x,y)}\langle v,gy\rangle.
\qedhere
\]
\end{proof}

\begin{theorem}
Suppose $G$ is a closed subgroup of $\operatorname{O}(d)$ and select $x,y\in\mathbb{R}^d$.
Then
\[
\partial\llangle [\cdot],[y]\rrangle(x)
=\operatorname{conv}\{gy:g\in G(x,y)\}.
\]
\end{theorem}

\begin{proof}
First, we claim that for every $g\in G(x,y)$, the vector $gy$ is in the subdifferential $\partial\llangle [\cdot],[y]\rrangle(x)$.
Indeed, for every $h\in\mathbb{R}^d$, we have
\[
\llangle [x+h],[y]\rrangle
\geq \langle x+h,gy\rangle
= \langle x,gy\rangle + \langle h,gy\rangle
= \llangle [x],[y]\rrangle + \langle h,gy\rangle,
\]
where the last step uses the fact that $g\in G(x,y)$.
Since $\partial\llangle [\cdot],[y]\rrangle(x)$ is convex, it follows that the $\supseteq$ portion of the desired result holds.
Next, suppose there exists $w\in\partial\llangle [\cdot],[y]\rrangle(x)$ such that $w\not\in\operatorname{conv}\{gy:g\in G(x,y)\}$.
Then there exists a hyperplane that separates $w$ from $\{gy:g\in G(x,y)\}$, i.e., there is a nonzero vector $v\in\mathbb{R}^d$ such that
\[
\langle v,gy\rangle
<\langle v,w\rangle
\]
for every $g\in G(x,y)$.
By Lemma~\ref{lem.directional derivative}, it follows that
\[
\llangle[\cdot],[y]\rrangle'(x;v)
=\max_{g\in G(x,y)}\langle v,gy\rangle
<\langle w,v\rangle
\leq \sup_{u\in\partial \llangle[\cdot],[y]\rrangle(x)}\langle u,v\rangle
=\llangle[\cdot],[y]\rrangle'(x;v),
\]
a contradiction.
This establishes the $\subseteq$ portion of the desired result.
\end{proof}

\subsection{Random templates and limit laws}

While the previous subsection was concerned with the differentiability needed to optimize a feature map, it is well known that random feature maps suffice for various tasks (e.g., Johnson--Lindenstrauss maps~\cite{JohnsonL:84} and random kitchen sinks~\cite{RahimiR:08}).
In fact, our bilipschitz result (Theorem~\ref{thm.main result}) uses random templates.
As such, one may be inclined to use random templates to produce features for classification.

In this subsection, we focus on the case in which $V=\mathbb{R}^d$ and $G\cong S_d$ is the group of $d\times d$ permutation matrices.
Consider a max filter bank consisting of independent standard gaussian templates $z_1,\ldots,z_n\in\mathbb{R}^d$.
Then
\[
\llangle [z_i],[x]\rrangle
=\langle \operatorname{sort}(z_i),\operatorname{sort}(x)\rangle.
\]
When $d$ is large, we expect the vectors $\operatorname{sort}(z_i)$ to exhibit little variation, and so we are inclined to perform dimensionality reduction.
Figure~\ref{fig.sorted gaussians} illustrates that the principal components of $\{\operatorname{sort}(z_i)\}_{i=1}^n$ exhibit a high degree of regularity.

\begin{figure}
\begin{center}
\includegraphics[width=0.45\textwidth,trim={110 210 110 210},clip]{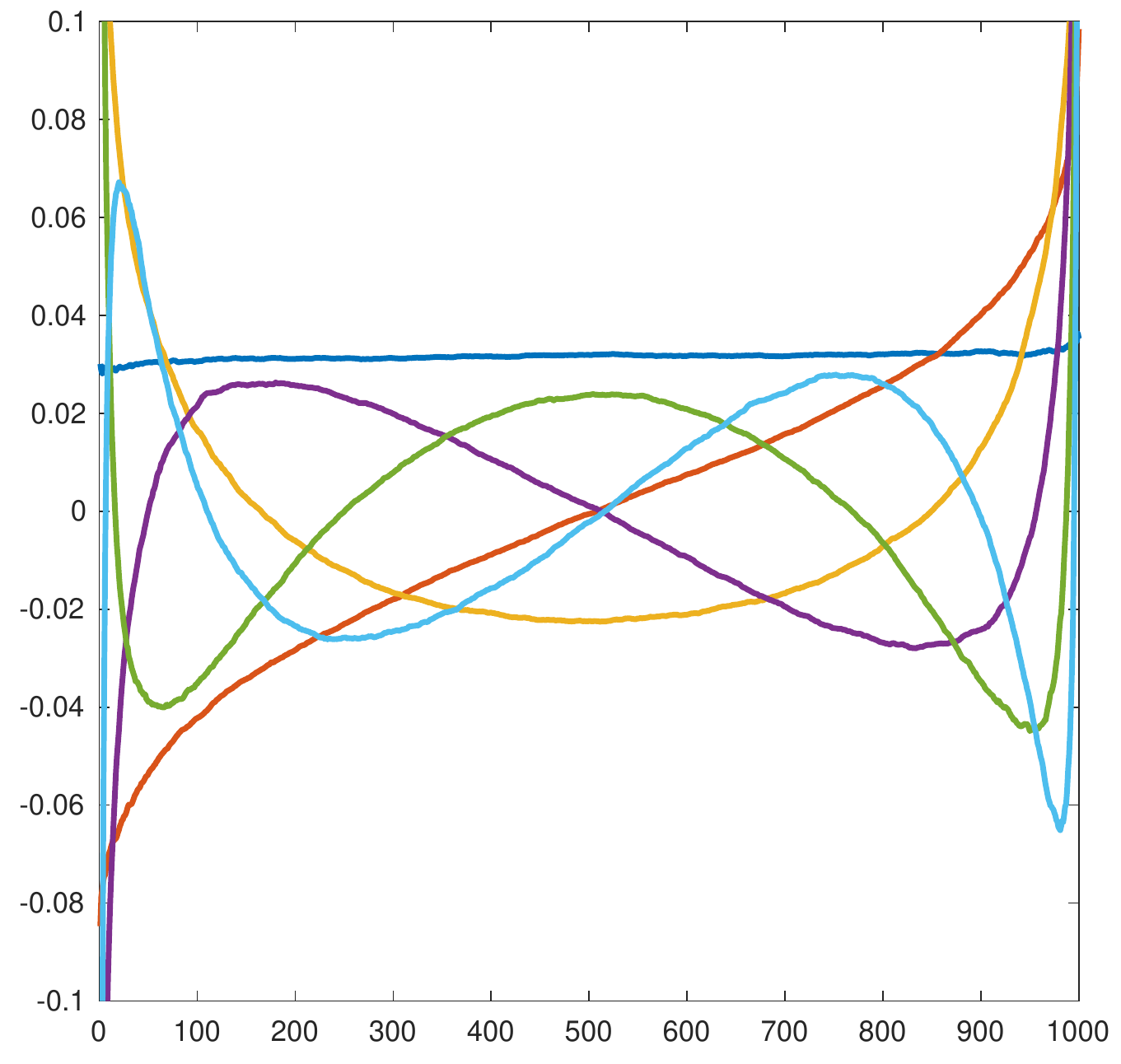}
\quad
\includegraphics[width=0.45\textwidth,trim={110 210 110 210},clip]{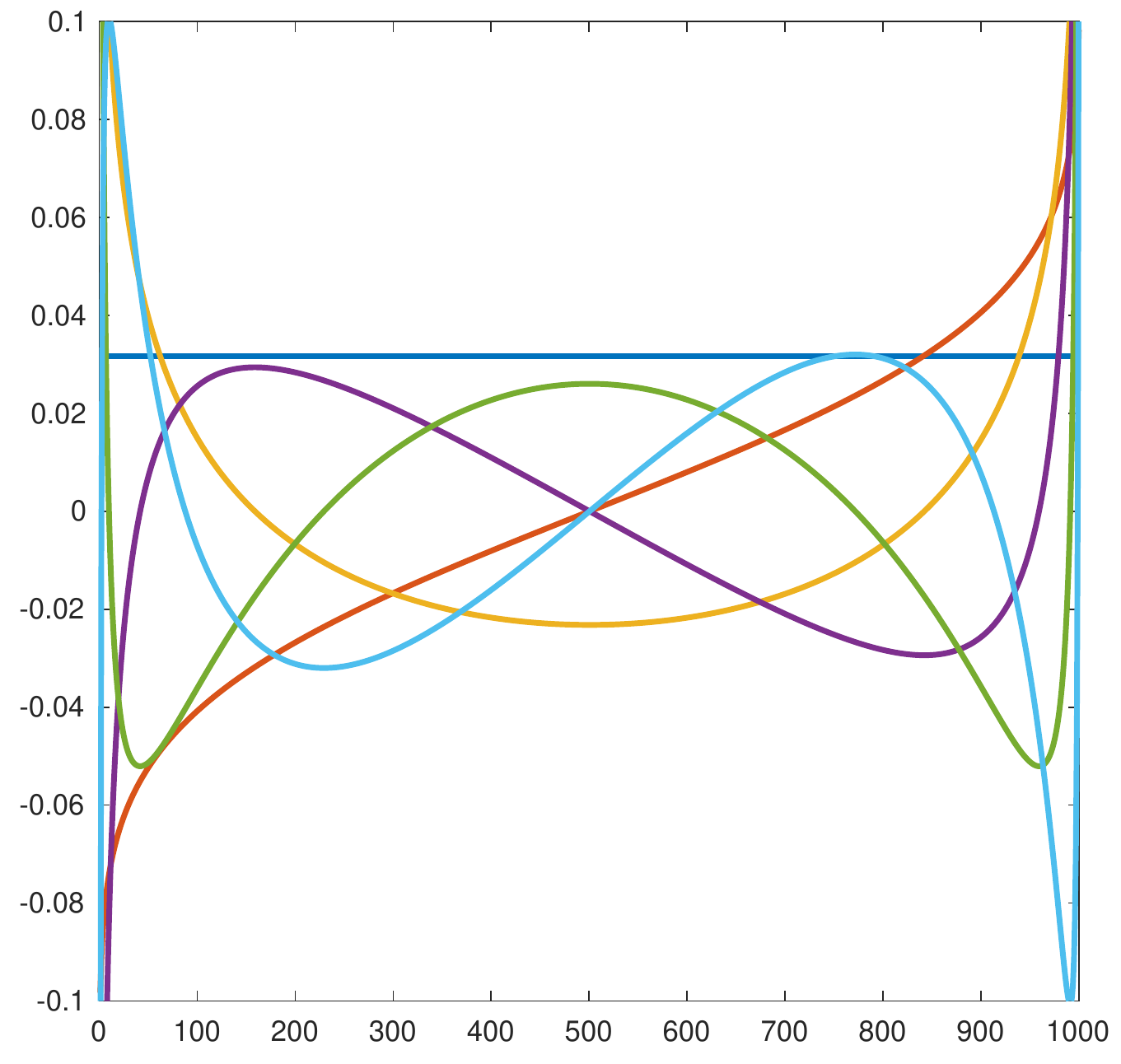}
\end{center}
\caption{\label{fig.sorted gaussians}
\textbf{(left)}
Draw $n=10,000$ independent standard gaussian random vectors in $\mathbb{R}^d$ with $d=1,000$, sort the coordinates of each vector, and then plot the top $6$ eigenvectors of the resulting sample covariance matrix.
\textbf{(right)}
Discretize and plot the top $6$ eigenfunctions described in Theorem~\ref{thm.eigenfunctions}.
}
\end{figure}

To explain this regularity, select $Q\colon(0,1)\to\mathbb{R}$ so that $Q^{-1}$ is the cumulative distribution function of the standard normal distribution.
Take $z\in\mathsf{N}(0,I_d)$ and put $s:=\operatorname{sort}(z)$.
Denoting $p_i:=\frac{i}{d+1}$ and $q_i:=1-p_i$, then Section~4.6 in~\cite{HerbertN:04} gives
\[
\mathbb{E}[s_i]=Q(p_i)+O(\tfrac{1}{d}),
\qquad
\operatorname{Cov}(s_i,s_j)=\tfrac{1}{d+2}p_iq_jQ'(p_i)Q'(p_j)+O(\tfrac{1}{d^{2}}),
\qquad
i\leq j.
\]
The principal components of $\{\operatorname{sort}(z_i)\}_{i=1}^n$ approximate the top eigenvectors of the covariance matrix, which in turn approximate discretizations of eigenfunctions of the integral operator with kernel $K\colon(0,1)^2\to\mathbb{R}$ defined by
\[
K(x,y):=\min\{x,y\}\cdot(1-\max\{x,y\})\cdot Q'(x)\cdot Q'(y).
\]
The following result expresses these eigenfunctions in terms of the \textit{probabilist's Hermite polynomials}, which are defined by the following recurrence:
\begin{equation}
\label{eq.hermite recurrence}
p_0(x)
:=1,
\qquad
p_{n+1}(x)
:=xp_n(x)-p_n'(x).
\end{equation}

\begin{theorem}
\label{thm.eigenfunctions}
The integral operator $L\colon L^2([0,1])\to L^2([0,1])$ defined by 
\[
Lf(x):=\int_0^1 K(x,y)f(y)dy
\]
has eigenvalue $\frac{1}{n+1}$ with corresponding eigenfunction $p_n\circ Q$ for each $n\in\mathbb{N}\cup\{0\}$.
\end{theorem}

As such, instead of max filtering with independent gaussian templates, one can efficiently capture the same information by taking the inner product between $\operatorname{sort}(x)$ and discretized versions of the eigenfunctions $p_n\circ Q$.
To reduce the dimensionality, one can simply use fewer eigenfunctions.
To prove Theorem~\ref{thm.eigenfunctions}, we will use the following lemma; here, $\varphi\colon\mathbb{R}\to\mathbb{R}$ denotes the probability density function of the standard normal distribution.

\begin{lemma}\
\label{lem.derivatives go down and up}
\begin{itemize}
\item[(a)]
$p_n(x)=\frac{1}{n+1}p_{n+1}'(x)$.
\item[(b)]
$\int_0^x p_{n+1}(Q(y))dy=-\varphi(Q(x))p_n(Q(x))$.
\end{itemize}
\end{lemma}

The proof of Lemma~\ref{lem.derivatives go down and up} follows quickly from the recurrence~\eqref{eq.hermite recurrence}.

\begin{proof}[Proof of Theorem~\ref{thm.eigenfunctions}]
We compute $L(p_n\circ Q)$ by splitting the integral:
\begin{align*}
L(p_n\circ Q)(x)
&= (1-x)Q'(x)\underbrace{\int_0^x yQ'(y)p_n(Q(y))dy}_{I_1}
+xQ'(x)\underbrace{\int_x^1 (1-y)Q'(y)p_n(Q(y))dy}_{I_2}.
\end{align*}
For both $I_1$ and $I_2$, we integrate by parts with
\[
dv
=Q'(y)p_n(Q(y))dy
=Q'(y)\cdot \tfrac{1}{n+1}p_{n+1}'(Q(y))\cdot dy
=\tfrac{1}{n+1}(p_{n+1}\circ Q)'(y) dy,
\]
where the middle step follows from Lemma~\ref{lem.derivatives go down and up}(a).
Then Lemma~\ref{lem.derivatives go down and up}(b) gives
\begin{align*}
I_1
&=yv\Big|_0^x-\int_0^x vdy
=\tfrac{1}{n+1}\Big(xp_{n+1}(Q(x))+\varphi(Q(x))p_n(Q(x))\Big),\\
I_2
&=(1-y)v\Big|_x^1+\int_x^1 vdy
=\tfrac{1}{n+1}\Big(-(1-x)p_{n+1}(Q(x))+\varphi(Q(x))p_n(Q(x))\Big).
\end{align*}
These combine (and mostly cancel) to give
\[
L(p_n\circ Q)(x)
=\tfrac{1}{n+1}Q'(x)p_n(Q(x))\varphi(Q(x))
=\tfrac{1}{n+1}(p_n\circ Q)(x).
\qedhere
\]
\end{proof}

\section{Numerical examples}
\label{sec.numerics}

In this section, we use max filters as feature maps for various real-world learning tasks.

\begin{example}[Voting districts]
\label{ex.districts}
The \textit{one person, one vote} principle insists that in each state, different voting districts must have nearly the same number of constituents.
This principle is enforced with the help of a decennial redistricting process based on U.S.\ Census data.
Interestingly, each state assembly applies its own process for redistricting, and partisan approaches can produce unwanted gerrymandering.
Historically, gerrymandering is detected by how contorted a district's shape looks; for example, the Washington Post article~\cite{Ingraham:14} uses a particular geometric score to identify the top $10$ most gerrymandered voting districts of the $113$th Congress.

As an alternative, we visualize the distribution of district shapes with the help of max filtering.
The shape files for all voting districts of the $116$th Congress are available in~\cite{Switzer:kaggle}.
We center each district at the origin and scale it to have unit perimeter, and then we sample the boundary at $n=50$ equally spaced points.
This results in a $2\times 50$ matrix representation of each district.
However, the same district shape may be represented by many matrices corresponding to an action of $\operatorname{O}(2)$ on the left and an action of $C_n$ that cyclically permutes columns.
Hence, it is appropriate to apply max filtering with $G\cong\operatorname{O}(2)\times C_n$.
It is convenient to identify $\mathbb{R}^{2\times 50}$ with $\mathbb{C}^{50}$ so that the corresponding max filter is given by
\begin{align*}
\llangle [z],[x]\rrangle
&=\max\Big\{~\max_{a\in C_n} |z^*T_ax|,~\max_{a\in C_n}| \overline{z}^*T_ax|~\Big\}\\
&=\max\Big\{~\max_{a\in C_n}|(R\overline{z}\star x)(a)|,~\max_{a\in C_n}|(Rz\star x)(a)|~\Big\},
\end{align*}
which can be computed efficiently with the help of the fast Fourier transform.
We max filter with $100$ random templates to embed these districts in the feature space $\mathbb{R}^{100}$, and then we visualize the result using PCA; see Figure~\ref{fig.districts}.

\begin{figure}
\begin{center}
\includegraphics[width=\textwidth]{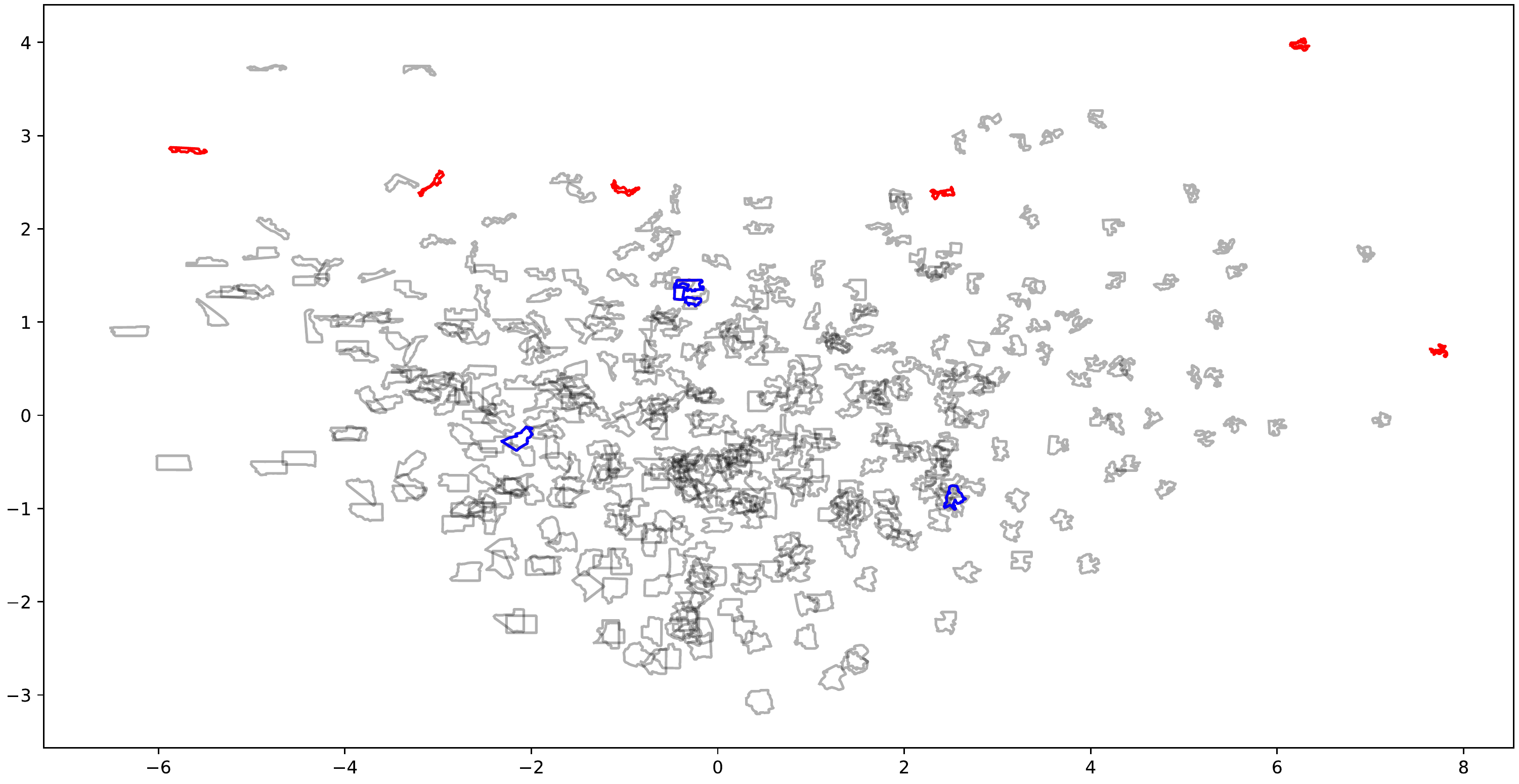}
\end{center}
\caption{\label{fig.districts}
Visualization of voting districts of the $116$th Congress obtained by max filtering and principal component analysis.
See Example~\ref{ex.districts} for details.
}
\end{figure}

Interestingly, the principal components of this feature domain appear to be interpretable.
The first principal component (given by the horizontal axis) seems to capture the extent to which the district is convex, while the second principal component seems to capture the eccentricity of the district.
Six of the ten most gerrymandered districts from~\cite{Ingraham:14} are drawn in red, which appear at relatively extreme points in this feature space.
The four remaining districts (NC-1, NC-4, NC-12, and PA-7) were redrawn by court order between the $113$th and $116$th Congresses; the new versions of these districts are drawn in blue.
Unsurprisingly, the redrawn districts are not contorted, and they appear at not-so-extreme points in the feature space.
\end{example}

\begin{example}[ECG time series]
\label{ex.ecg}
An electrocardiogram (ECG) uses electrodes placed on the skin to quantify the heart's electrical activity over time.
With ECG data, a physician can determine whether a patient has had a heart attack with about 70\% accuracy~\cite{MakimotoEtal:20}.
Recently, Makimoto et al.~\cite{MakimotoEtal:20} trained a $6$-layer convolutional neural network to perform this task with about 80\% accuracy.
The dataset they used is available at~\cite{BousseljotKS:95}, which consists of $(12+3)$-lead ECG data sampled at $1$~kHz from $148$ patients who recently had a heart attack (i.e., \textit{myocardial infarction} or \textit{MI}) and from $141$ patients who did not.

As an alternative to convolutional neural networks, we use max filters to classify patients based on \textit{one second} of ECG data (i.e., roughly one heartbeat).
In particular, we read the first $t=1000$ samples of all $15$ leads to obtain a $15\times t$ matrix $X$.
We then lift this matrix to a $15\times w\times (t-w+1)$ tensor with $w=30$, where each $15\times w$ slice corresponds to a contiguous $15\times w$ submatrix of $X$.
Then we normalize each $1\times w\times 1$ subvector to have mean zero to get the tensor $Y$; this discards any trend in the data.
We account for time invariance by taking $G$ to be the order-$(t-w+1)$ group of circular permutations of the $15\times w$ slices of $Y$.
Using this group, we max filter with $n=5$ different templates and then classify with a support-vector machine. 
We constrain each template to be supported on a single slice, and we train these templates together with the support-vector machine classifier by gradient descent to minimize hinge loss.

\begin{figure}
\begin{center}
\includegraphics[width=\textwidth]{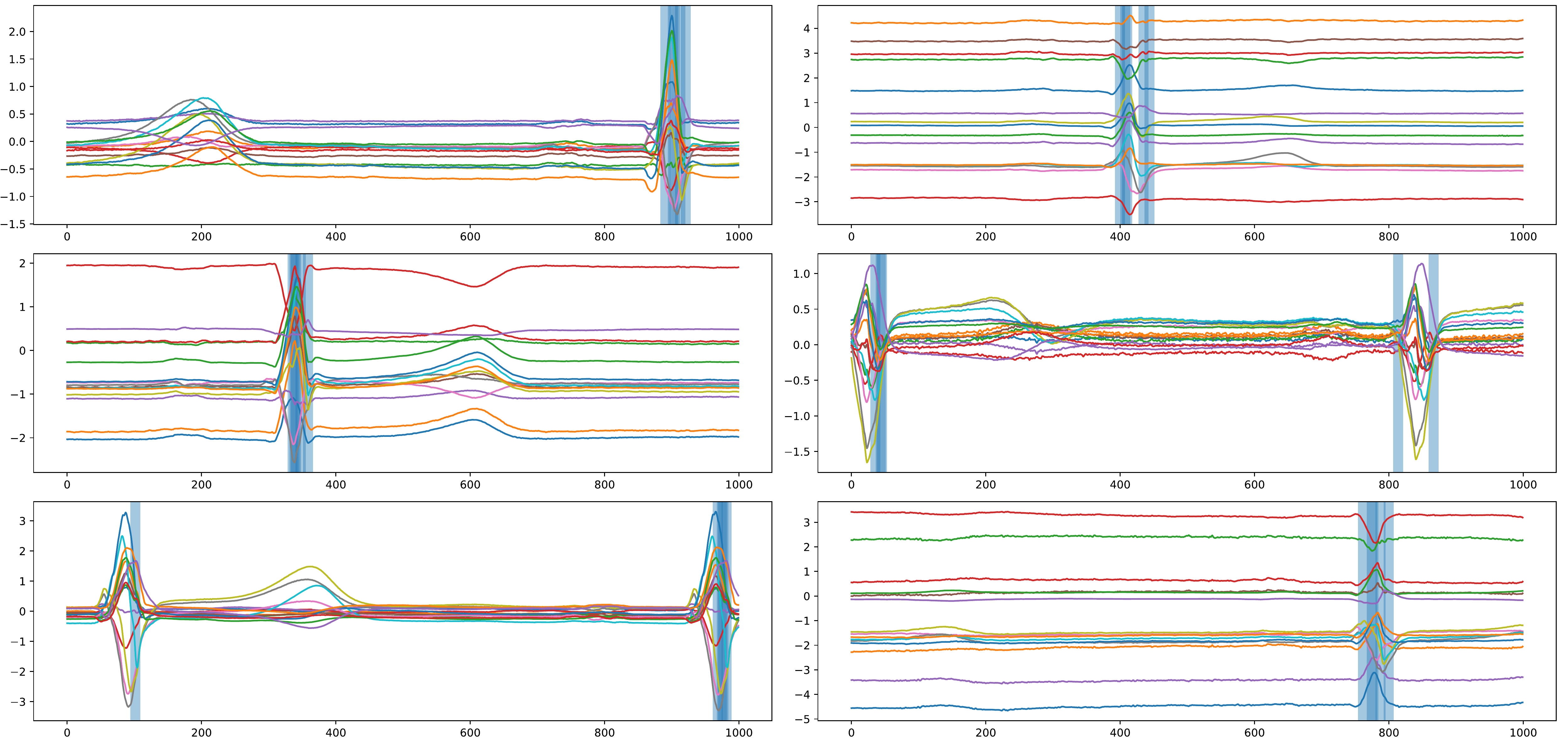}
\end{center}
\caption{\label{fig.ecg}
We train max filter templates and a support-vector machine classifier on electrocardiogram data to distinguish between patients who have had a heart attack from those who have not.
Above, we plot the most extreme examples in the test set (those with heart attacks on the left, and those without on the right).
The blue windows illustrate the time segments of width $w=30$ that align best with the templates.
See Example~\ref{ex.ecg} for details.
}
\end{figure}

Following~\cite{MakimotoEtal:20}, we form the test set by randomly drawing $25$ examples from the MI class and $25$ examples from the non-MI class, and then we form the training set by randomly drawing $108$ non-test set examples from each class.
We perform $10$ trials of this experiment, and each achieve between $74$\% and $84$\% accuracy on the test set.
In particular, this is competitive with the $6$-layer convolutional neural network in~\cite{MakimotoEtal:20} despite accessing only a fraction of the data.
In Figure~\ref{fig.ecg}, we illustrate what the classifier considers to be the most extreme examples in the test set.
The time segments that align best with the trained templates typically cover the heartbeat portion of the ECG signal.
\end{example}

\begin{example}[Textures]
\label{ex.textures}
In this example, we use max filtering to classify various textures, specifically, those given in the Kylberg Texture Dataset v.\ 1.0, available in~\cite{Kylberg:online}.
This dataset consists of 28 classes of textures, such as blanket, grass, and oatmeal.
Each class consists of 160 distinct grayscale images.
(We crop these $576\times 576$ images down to $256\times 256$ for convenience.)
A few classifiers in the literature have been trained on this dataset~\cite{KylbergS:13,AndrearczykW:16}, but in order to achieve a high level of performance, they leverage domain-specific feature maps, augment the training set with other image sets, or simply initialize with a pre-trained network.
As an alternative, we apply max filtering with random features, and we succeed with much smaller training sets.
For an honest comparison, we train different classifiers on training sets of different sizes, and the results are displayed in Table~\ref{table.accuracies}.
(Here, the test set consists of $32$ points from each class.)
We describe each classifier in what follows.

\begin{table}
\caption{\label{table.accuracies}
Accuracy on test set versus size of training set (baseline = 0.035, see Example~\ref{ex.textures})
}
\begin{center}
\begin{tabular}{|>{\centering\arraybackslash}p{0.17\textwidth}>{\centering\arraybackslash}p{0.17\textwidth}>{\centering\arraybackslash}p{0.17\textwidth}>{\centering\arraybackslash}p{0.17\textwidth}>{\centering\arraybackslash}p{0.17\textwidth}|}\hline
pts per class & CNN & LDA & PCA-LDA & our method \\ \hline\hline
2 & 0.125 & 0.043 & 0.063 & 0.698 \\ \hline
4 & 0.114 & 0.103 & 0.060 & 0.947 \\ \hline
8 & 0.150 & 0.158 & 0.128 & 0.970 \\ \hline
16 & 0.214 & 0.125 & 0.128 & 0.970 \\ \hline
32 & 0.306 & 0.116 & 0.120 & 0.968 \\ \hline
64 & 0.534 & 0.117 & 0.122 & 0.964 \\ \hline
128 & 0.637 & (OOM) & 0.128 & 0.962 \\ \hline
\end{tabular}
\end{center}
\end{table}

First, we consider a convolutional neural network (CNN) with a standard architecture.
We pass the $256\times 256$ image through a convolutional layer with a $3\times 3$ kernel before max pooling over $2\times 2$ patches.
Then we pass the result through another convolutional layer with a $3\times 3$ kernel, and again max pool over $2\times 2$ patches.
We end with a dense layer that maps to $28$ dimensions, one for each class.
We train by minimizing cross entropy loss against the one-hot vector representation of each label.
This model has $1.7$ million parameters, which enables us to achieve perfect accuracy on the training set.

As an alternative, we apply linear discriminant analysis (LDA).
Training this classifier involves a massive matrix operation, and this requires too much memory in our setting when the training set has $128$ points per class.
This motivates the use of principal component analysis to reduce the dimensionality of each $256\times 256$ image to $k=25$ principal components.
The resulting classifier (PCA-LDA) works for larger training sets, and even exhibits improved accuracy when the training set is not too small.

For our method, we apply the same PCA-LDA method to a max filtering feature domain.
Surprisingly, we find that modding out by the entire group of pixel permutations (i.e., just sorting the pixel values) performs reasonably well as a feature map.
Our method implements a multiscale version of this observation.
For each $\ell\in\{0,1,\ldots,8\}$, one may partition the $2^8\times 2^8$ image into $4^{8-\ell}$ square patches of width $2^\ell$.
We perform max filtering with the group $(S_{4^\ell})^{4^{8-\ell}}$ of all patch-preserving pixel permutations for each $\ell\in\{2,\ldots,8\}$.
For simplicity, we are inclined to draw templates at random, but for computational efficiency, we simulate random templates with the help of Theorem~\ref{thm.eigenfunctions}.
In particular, for each $\ell\in\{2,\ldots,8\}$ and $n\in\{0,\ldots,5\}$, we sort the pixels in each $2^\ell\times 2^\ell$ patch and take their inner products with a discretized version of $p_n\circ Q$ and sum over the patches.
This maps each $2^8\times 2^8$ image to a $42$-dimensional feature vector.
As before, we apply PCA with $k=25$ and then train an LDA classifier on the result.
In the end, this max filtering approach significantly outperforms the other classifiers we considered, as seen in Table~\ref{table.accuracies}.
\end{example}

\section{Discussion}
\label{sec.discussion}

Max filtering offers a rich source of invariants that can be used for a variety of machine learning tasks.
In this section, we discuss several opportunities for follow-on work.

\textbf{Separating.}
In terms of sample complexity, Corollary~\ref{cor.2d templates suffice for finite groups} gives that a generic max filter bank of size $2d$ separates all $G$-orbits in $\mathbb{R}^d$ provided $G\leq\operatorname{O}(d)$ is finite.
If $G\leq\operatorname{O}(d)$ is not topologically closed, then no continuous invariant (such as a max filter bank) can separate all $G$-orbits.
If $G\leq\operatorname{O}(d)$ is topologically closed, then by Theorem~3.4.5 in~\cite{OnishchikV:90}, $G$ is algebraic, and so Theorem~\ref{thm.dym-gortler} applies.
We suspect that max filtering separates orbits in such cases, but progress on this front will likely factor through Problem~\ref{prob.semialgebraic}(a).
For computational complexity, how well does the max filtering separation hierarchy~\eqref{eq.hierarchy} separate isomorphism classes of weighted graphs?
Judging by~\cite{AlonYZ:95}, we suspect that max filtering with a template of order $k$ and treewidth $t$ can be computed with runtime $e^{O(k)}n^{t+1}\log n$, which is polynomial in $n$ when $k$ is logarithmic and $t$ is bounded.
Which classes of graphs are separated by such max filters?
Also, can max filtering be used to solve real-world problems involving graph data, or is the exponent too large to be practical?

\textbf{Bilipschitz.}
The proof of Lemma~\ref{lem.bilipschitz no randomness} contains our approach to finding Lipschitz bounds on max filter banks.
Our upper Lipschitz bound follows immediately from the fact that each individual max filter is Lipschitz, and it does not require the group to be finite.
This bound is optimal for $G=\operatorname{O}(d)$, but it is known to be loose for small groups like $G=\{\pm\operatorname{id}\}$.
We suspect that our lower Lipschitz bound leaves a lot more room for improvement.
In particular, we use the pigeonhole principle to pass from a sum to a maximum so that we can leverage projective uniformity.
A different approach might lead to a tighter bound that does not require $G$ to be finite.
More abstractly, we suspect that separating implies bilipschitz, though the bounds might be arbitrarily bad; see Problem~\ref{prob.sep implies bilip}.

\textbf{Randomness.}
Our separating and bilipschitz results (Corollary~\ref{cor.2d templates suffice for finite groups} and Theorem~\ref{thm.main result}) are not given in terms of explicit templates.
Meanwhile, our Mallat-type stability result (Theorem~\ref{thm.mallat bound}) requires a localized template, and we can interpret the templates used in our weighted graph separation hierarchy~\eqref{eq.hierarchy} as being localized, too.
We expect that there are other types of structured templates that would reduce the computational complexity of max filtering (much like the structured measurements studied in compressed sensing~\cite{RudelsonV:08,PfanderRT:13,KrahmerMR:14} and phase retrieval~\cite{BandeiraCM:14,EldarSMBC:14,BodmanH:15,GrossKK:17}).
It would be interesting to prove separating or bilipschitz results in such settings.
More generally, can one construct explicit templates for which max filtering is separating or bilipschitz for a given group?
Going the other direction, the reader may have noticed that the plots in Figure~\ref{fig.sorted gaussians} deviate from each other at the edges of the interval.
We expect that such deviations decay as the dimension grows, but this requires further analysis.
How can one analyze the behavior of random max filters for other groups?

\textbf{Max filtering networks.}
Example~\ref{ex.textures} opens the door to a variety of innovations with max filtering.
Here, instead of fixing a single group to mod out by, we applied max filtering with a family of different groups.
We selected permutations over patches due to the computational simplicity of using Theorem~\ref{thm.eigenfunctions}, and we arranged the patches in a hierarchical structure so as to capture behavior at different scales.
Is there any theory to explain the performance of this architecture?
Are there other useful architectures in this vicinity, perhaps by combining with neural networks?

\section*{Acknowledgments}

This work was initiated at the SOFT 2021:\ Summer of Frame Theory virtual workshop.
The authors thank Boris Alexeev for helpful discussions and Soledad Villar for bringing the article~\cite{DymG:22} to their attention.

\end{document}